\documentclass[11pt,reqno]{amsart}

\usepackage{amsfonts}
\usepackage{amsmath}
\usepackage{amssymb}
\usepackage{latexsym}
\usepackage{graphicx} 
\usepackage[table]{xcolor}
\usepackage{cite}
\usepackage{enumerate}
\usepackage[colorlinks=true,urlcolor=black,linkcolor=black,citecolor=black,pagebackref=black,]{hyperref}

\numberwithin{equation}{section}

\newtheorem{thm}[equation]{Theorem}

\newtheorem{lem}[equation]{Lemma}

\numberwithin{equation}{section}

\renewcommand\a{\alpha}

\newcommand\g{\gamma}
\renewcommand\d{\delta}
\newcommand\e{\varepsilon}
\renewcommand\l{\lambda}
\renewcommand\L{\Lambda}

\newcommand\G{\Gamma}
\newcommand\f{\frac}
\newcommand\smallf[2]{{\textstyle{\frac{#1}{#2}}}}
\newcommand\srel[2]{\begin{smallmatrix} {#1} \\ {#2} \end{smallmatrix}}

\newcommand{\Z}{{\mathbb{Z}}}
\newcommand{\R}{{\mathbb{R}}}

\newcommand{\C}{{\mathbb{C}}}

\newcommand{\Q}{{\mathbb{Q}}}

\newcommand{\U}{{\mathbb{H}}}

\newcommand\im{\mbox{Im~}}
\renewcommand\Re{\text{Re~}}

\renewcommand\({\left(}
\renewcommand\){\right)}
\newcommand{\ttwo}[4]{
\(\begin{smallmatrix}{#1} & {#2}
\\ {#3} & {#4} \end{smallmatrix}\)}

\newcommand{\sgn}{\operatorname{sgn}}

\newcommand{\gobble}[1]{}
  \newcommand{\rangeref}[2]{%
~\ref{#1}--\afterassignment\gobble\fam 0\ref{#2}%
  }
\makeatletter
\def\imod#1{\allowbreak\mkern10mu({\operator@font mod}\,\,#1)}
\makeatother

\def\IR{{\mathbb R}}

\def\EZ(#1,#2){\mathcal{Z}_{\left(#1,#2\right)}}

\def\cE{{\cal E}}

\def\cE{\mathcal{E}}
\def\cF{\mathcal{F}}
\def\cV{\mathcal{V}}

\newcommand{\be}{\begin{equation}}
\newcommand{\ee}{\end{equation}}
\newcommand{\bea}{\begin{eqnarray}}
\newcommand{\eea}{\end{eqnarray}}

\def\a{\alpha }
\def\g{\gamma }

\def\threeh{{\scriptstyle {3 \over 2}}}
\def\fiveh{{\scriptstyle {5 \over 2}}}

\def\calE{{\mathcal E}}

\def\calT{{\mathcal T}}

\def\calF{{\mathcal F}}

\def\nn{\nonumber}
\def\half{{\scriptstyle {1 \over 2}}}

\def\quart{{\scriptstyle {1 \over 4}}}
\def\eighth{{\scriptstyle {1 \over 8}}}
\def\sevenh{{\scriptstyle {7 \over 2}}}
\def\eighth{{\scriptstyle {1\over 8}}}
\title[$SL(2,\Z)$-invariance and D-instanton contributions]{$SL(2,\Z)$-invariance and D-instanton
contributions to the $D^6R^4$ interaction}

\author[M.B. Green]{Michael B. Green}
\address{Michael B. Green\\
Department of Applied Mathematics and
Theoretical Physics\\
 Wilberforce Road, Cambridge CB3 0WA, UK}
\email{ M.B.Green@damtp.cam.ac.uk}
\author[S.D. Miller]{Stephen D. Miller}
\address{Stephen D Miller\\
Department of Mathematics\\
 Rutgers University, Piscataway, NJ 08854-8019, USA}
\email{miller@math.rutgers.edu}
\author[P. Vanhove]{Pierre Vanhove}
 \address{Pierre Vanhove\\
Institut des Hautes Etudes Scientifiques\\
 Le Bois-Marie, 35 route de Chartres\\
 F-91440 Bures-sur-Yvette, France\hfill\break
Institut de Physique Th{\'e}orique,\\
CEA, IPhT, F-91191 Gif-sur-Yvette, France\\
CNRS, URA 2306, F-91191 Gif-sur-Yvette, France}
\email{pierre.vanhove@cea.fr}

\thanks{DAMTP-2014-14   , IPHT-T-14/009, IHES/P/14/07}
\date{}

\begin{document}

 \begin{abstract}
   The modular invariant coefficient of the $D^6R^4$ interaction in
   the low energy expansion of type~IIB string theory has been
   conjectured to be a solution of an inhomogeneous Laplace eigenvalue
   equation, obtained by considering the toroidal compactification of
   two-loop Feynman diagrams of eleven-dimensional supergravity.  In
   this paper we determine the exact $SL(2,\Z)$-invariant solution
   $f(x+iy)$ to this differential equation
   satisfying an appropriate moderate growth condition as $y\to
   \infty$ (the weak coupling limit).  The solution is presented as a
   Fourier series with modes $\widehat{f}_n(y) e^{2\pi i n x}$, where
   the mode coefficients, $\widehat{f}_n(y)$ are bilinear in
   $K$-Bessel functions.  Invariance under $SL(2,\Z)$ requires these
   modes to satisfy the nontrivial boundary condition $
   \widehat{f}_n(y) =O(y^{-2})$ for small $y$, which uniquely
   determines the solution. The large-$y$ expansion of $f(x+iy)$
   contains the known perturbative (power-behaved) terms, together
   with precisely-determined exponentially decreasing contributions
   that have the form expected of D-instantons, anti-D-instantons and
   D-instanton/anti-D-instanton pairs.

\end{abstract}

\maketitle
\tableofcontents


\section{Introduction}
\label{intro}

The low energy expansion of string theory has a rich dependence on the
moduli, or scalar  fields, that parameterize the coset space
$G(\R)/K(\R)$, where $G$ is the duality group and $K$ its maximal
compact subgroup.  In this paper we will be concerned with the
simplest nontrivial example, type IIB superstring theory in $D=10$
space-time dimensions,  in which $G= SL(2)$ and $K = SO(2)$.  Duality
invariance of the theory  implies that the IIB scattering amplitudes
should transform covariantly  under the discrete arithmetic subgroup,
$G(\Z) = SL(2,\Z)$.   This implies that the coefficients of the
terms at any order in the low energy expansion of the amplitude are
modular functions, which restricts their dependence on the moduli.

Terms of sufficiently low dimension in the effective action
preserve a fraction of the 32 supercharges, i.e., they are BPS
interactions.  Such interactions have particularly simple
moduli-dependent coefficients.  The lowest-order terms that contribute to the four-particle amplitude,  beyond the
Einstein--Hilbert action, are the $\half$-BPS and $\quart$-BPS
interactions of order $R^4$ and $D^4R^4$,  where $R^4$ denotes four
powers of the Riemann curvature tensor with the sixteen indices
contracted by a standard sixteen-index tensor  \cite[appendix~9.A]{Green:1987mn}  that will not concern us
here.  These interactions have  coefficients given by non-holomorphic
Eisenstein series, $E_\threeh(\Omega)$ and $E_{\fiveh}(\Omega)$,
respectively   (we refer to~\eqref{eisenfourier} for a definition of these series).  Here $\Omega=x+iy$ is the complex modulus and
$y^{-1}=g_B$ is the type IIB string coupling.

It is the coefficient of the next term, the $\eighth$-BPS interaction
$D^6R^4$, that is the subject of this paper.  This interaction enters into the type IIB string frame low energy effective action in the form
 \be
\ell_s^4 \int d^{10}x\sqrt{-\det G^{(10)}}\, y^{-1}\, f(\Omega)\, D^6R^4\,,
 \label{effaction}
\ee
where $G^{(10)}$ is the ten-dimensional string frame metric, $\ell_s$ is the string length scale and we have suppressed an overall numerical coefficient.  The factor of $y^{-1}$ cancels when $G^{(10)}$ is rescaled in a manner that converts the expression to the Einstein frame, in which $SL(2,\Z)$ duality should be manifest.  The coefficient $f(\Omega)$ is a modular function  that was conjectured in \cite{Green:2005ba} to be the solution of an
inhomogeneous Laplace eigenvalue equation\footnote{In   reference \cite{Green:2005ba}, $f(\Omega)$ was denoted  by $\calE_{(\threeh,\threeh)}$; it has also been denoted by $\calE_{(0,1)}$ in earlier papers on this subject, such as \cite{GMRV} and \cite{GMV}.}
\be
(\Delta_\Omega- 12)\, f(\Omega)  \ \ = \ \  - \(2\,\zeta(3)\,E_{\threeh}(\Omega)\)^2 ,
\label{laplaceeigen}
\ee
where $\Delta_\Omega =y^2\, \left( \partial^2_{x} + \partial_{y}^2
\right)$.
The basis of this conjecture was an implementation of the duality that
relates M-theory compactified on a torus to type~IIB string theory
compactified on a circle~\cite{Witten:1995ex,Schwarz:1995dk,Aspinwall:1995fw}.  More precisely, the procedure used in
\cite{Green:2005ba} was to evaluate the terms of order $D^6\, R^4$  in
two-loop four-graviton supergravity amplitude compactified on a
two-torus to nine dimensions.  The complex structure of the torus, $\Omega$,
translates into the complex coupling constant of the type~IIB  string theory while the torus volume, $\cV$,  is proportional
to an inverse power of the radius of the string theory circle, $r_B$.

However, the analysis of~\eqref{laplaceeigen}  in
\cite{Green:2005ba} was incomplete in several respects. Although the
power-behaved terms in the large-$y$ expansion of $f(\Omega)$ were
determined in~\cite{Green:2005ba}, a general analysis of the solution
including the non power-behaved parts of the solution was missing.
The objective of this paper is to develop such an analysis.
Furthermore, since~\eqref{laplaceeigen}  is the simplest of  the more
general inhomogeneous eigenvalue equations that arise at higher orders
in the low energy expansion~\cite{Green:2008bf}, such an analysis
should be of more general significance.

The layout of this paper is as follows.    The detailed solution of~\eqref{laplaceeigen} is given in section~\ref{sec:fouriertransform}.  Our procedure is to consider the inhomogeneous second order differential equations satisfied by the mode coefficients of the Fourier series
\be
f(\Omega) \ \  = \ \  \sum_n \widehat{f}_n(y)\, e^{2\pi i n x}\,.
\label{fouriermodes}
\ee
This requires the imposition of appropriate boundary conditions on
$\widehat{f}_n(y)$  at $y\to \infty$ and $y\to 0$.  The $y\to \infty$
condition (the weak coupling limit)  is determined by the moderate
growth condition\footnote{In the present context, this condition means that for
any $y_0>0$ there exists some constant $C>0$ such that $|f(x+iy)|\le C y^3$
 for all $y\ge y_0$.} requiring that
\be
\label{largebound}
 f(\Omega) \ \  = \ \  O(y^3)\,,
\ee
which corresponds to tree-level behaviour of the $D^6R^4$ interaction in string perturbation
theory.
 The $y\to 0$ condition  (the strong coupling limit),  which is much less obvious,  requires
\be
  \widehat{f}_n (y) \ \ =  \ \  O(y^{-2})\,.
\label{yzerocon}
\ee
 We will see that this condition  follows from a subtle relation
 between the weak coupling limit condition and  $SL(2,\Z)$
 invariance.
 These boundary conditions
pin down the solution completely (with no arbitrary  undetermined  coefficients)  and
we are able to determine the exact solution for $\widehat{f}_n(y)$ for all $n$:

\begin{thm}\label{mainthm}
The unique solution    (\ref{laplaceeigen}) satisfying (\ref{largebound}) is given in terms of (\ref{fouriermodes}) by
\bea
\label{summarysol}
\widehat{f}_n(y) &  =&  \delta_{n,0}\, \tilde f(y) \ + \  \alpha_{n}\,\sqrt{y}\, K_{\sevenh}(2\pi |n|y) \\
&&+\sum_{\srel{n_1+n_2\,=\,n}{(n_1,n_2)\,\ne\, (0,0)}}
\sum_{i,j\,=\,0,1}M^{ij}_{n_1,n_2}(\pi |n| y)\, K_i(2\pi |n_1|y)\, K_j(2\pi |n_2| y) \nn\,,
\eea
where $\alpha_n$ are constants, and $\tilde f$ and   $M^{ij}_{n_1,n_2}(z)$  are  polynomials in $z$ and $1/z$.  The $K$-Bessel functions must be replaced by an appropriate limit when either $n$, $n_1$, or $n_2$ vanishes; see section~\ref{sec:solutions} for complete details.
\end{thm}

As we shall now explain, each of the coefficients in the solution (\ref{summarysol}) has an interpretation in terms of quantities arising in string theory in the limit $y\to \infty$.  The parameter $y$ is the inverse string coupling constant and so this limit is the weak coupling limit, in which the dominant terms in the solution are power behaved in $y^{-1}$ and correspond to contributions in string perturbation theory.  Such terms arise in the $\tilde f(y)$ part of the $n=0$ mode, which has the form $\tilde f(y) = a_0y^3 + a_1 y +a_2y^{-1}$.  The values of the coefficients $a_0$, $a_1$ and $a_2$ are rational numbers multiplied by products of zeta values.  These coefficients are expected to correspond to the values obtained from an analysis of the low energy expansion of superstring  perturbation theory.  In that context the coefficient, $a_h$, of $y^{3-2h}$  arises as a term in  the low energy expansion of the contribution of a genus-$h$ Riemann surface to superstring perturbation theory  (see the review article by D'Hoker and Phong \cite{D'Hoker:1988ta}, and references therein,  for a comprehensive description of string perturbation theory, and \cite{D'Hoker:2001nj,Berkovits:2005ng} for details of the
expression for the genus-two four-graviton amplitude).   One  additional power-behaved term arises in~\eqref{summarysol} from the
$n=0$ contribution $\lim_{n\to 0}\alpha_{n} \sqrt{y} \,
K_{\sevenh}(2\pi |n|y) = a_3 y^{-3}$, which should correspond to a genus-three contribution in string perturbation
theory.  The coefficients $a_0$, $a_1$, $a_2$, and $a_3$ were extracted from the constant term of $f(\Omega)$
by somewhat different means in~\cite{Green:2005ba}.  The present status of the comparison of these values with those obtained from superstring perturbation theory will be given in section \ref{sec:stringy}.

The large-$y$ behaviour of the $K$-Bessel functions
in~\eqref{summarysol} gives a rich spectrum of exponentially
decreasing terms that may be interpreted as D-instanton effects in
string theory\footnote{The terminology is motivated by the fact that
  large-$y$ behaviour proportional to    $e^{2\pi i (n_1+n_2)x}e^{-2\pi
    (|n_1|+|n_2|)y}$ is characteristic of contributions of
  D-instantons and anti-D-instantons, although the precise form of
  such contributions has not been obtained by explicit D-instanton
  calculations.}.  It is particularly notable that there are
instanton/anti-instanton terms in the large-$y$ expansion.  For
example, the zero mode, $\widehat{f}_0$, contains a sum of an infinite
series of exponentially suppressed terms of the form $
\sum_{n_1=1}^\infty c_{n_1}\, e^{-4\pi |n_1| y}$, where the
coefficients $c_{n_1}$ are easily deduced from the large-$y$ limit
of~\eqref{summarysol} as we will also describe in
section~\ref{sec:fouriertransform}.

 In section~\ref{sec:stringy} we will discuss how the  information in
 the solution of~\eqref{laplaceeigen} makes contact with string
 theory.  In particular, the small coupling (equivalently, large-$y$)  expansion of
 the solution obtained in section~~\ref{sec:fouriertransform} contains
 a rich array of instanton and anti-instanton contributions.  One of
 the main new observations in this paper is that these  conspire to
 ensure that the strong coupling ($y\to 0$) limits of  the Fourier
 modes satisfy the appropriate small-$y$ boundary condition.  This
 appears  somewhat analogous to the manner in
 which instanton effects conspire to ensure the absence of a
 singularity in three-dimensional  ${\mathcal N} =4$ supersymmetric
Yang--Mills theory in the work of Seiberg--Witten
 \cite{Seiberg:1996nz}.

For completeness, we will present several alternative procedures for determining the solution   to~\eqref{laplaceeigen}  in three appendices.
In  appendix~\ref{sec:poincaremeth},  we will make the $SL(2,\Z)$
properties  of~\eqref{laplaceeigen} explicit  by expressing the
solution as  a  series of the form
\begin{equation}\label{1.7}
  f(\Omega) \ \ = \ \ {2\,\zeta(3)^2\over3}\, E_3(\Omega)  \ + \  \sum_{\gamma\,\in\,\mathcal
    S} (\det \gamma)^{-3}\, F(\gamma\Omega)\,,
\end{equation}
with ${\mathcal S}=\{\pm 1\}\backslash \{\ttwo{m_1}{n_1}{m_2}{n_2}\in M_2(\Z)\cap GL^+(2,\IR)|\gcd(m_1,n_1)=\gcd(m_2,n_2)=1\}$ (which is the set of $2\times 2$  matrices with integer entries and co-prime rows modulo an overall sign).
The function $F(\Omega)$  depends only on the ratio of the real and imaginary parts of $\Omega$, and  satisfies a
second order inhomogeneous ordinary differential equation  given in~\eqref{poincare4}.
 The convergence of the sum over the images of
$F(\Omega)$ under $SL(2,\Z)$ transformations is obtained only if one imposes
 suitable boundary conditions  in the
 limits $x/y \to 0$ and $x/y\to \infty$.\footnote{We are grateful to
  Don Zagier for describing the solution satisfying the appropriate
  boundary conditions, as well as for discussions concerning the
  relevance of this solution.}
    The Fourier modes of the
   $SL(2,\Z)$-invariant expression (\ref{1.7}) are considered in
appendix~\ref{sec:fouriermeth}, where we  give an alternative expression of  the Fourier modes $\widehat{f}_n(y)$   of
   $f(\Omega)$  in terms of integrals.  We have not succeeded in directly computing those integrals, but their values are of course  determined by the  analysis of section~\ref{sec:fouriertransform}.  Furthermore, the convergence properties of these integrals again leads to the $y\to 0$ boundary condition that was deduced by general arguments in section~\ref{sec:fouriertransform}.

In appendix~\ref{sec:schmidmethod} we will describe how the solution may be obtained  in a manner suggested by Schmid's  work on automorphic distributions of Eisenstein series~\cite{flato,korea}.   This gives yet another formula for $\widehat{f}_n(y)$ in lemma~\ref{kloostermanfourierexpansion}.  In appendix~\ref{sec:spectral} we will comment on the solution using the
R\"olcke-Selberg  spectral expansion.  This leads to a complete solution
of~\eqref{laplaceeigen}, but one which seems to be very difficult to use in
practice (at least for the nonzero Fourier modes) since it involves properties of unknown cusp forms.

\section{Fourier modes of the inhomogeneous Laplace equation}
\label{sec:fouriertransform}

\subsection{ Fourier modes and boundary conditions}
\label{sec:boundaryterms}

We will now   consider  \eqref{laplaceeigen} in terms of the Fourier modes of both sides.
We write the Fourier expansion of the solution as
\begin{equation}
  f(x+iy)   \ \ = \ \  \sum_{n\,\in\,\Z} \, \widehat{f}_n(y)\,
  e^{2\pi i nx}
\label{fatstartofsect3}
\end{equation}
and the Fourier expansion of the source term  as
\bea
 S(x+iy)   \ \ = \ \   - 4\,\zeta(3)^2\,E_{3\over2}(x+iy)^2      \ \  = \ \   \sum_{n\,\in\,\Z}S_n(y)\,e^{2\pi i n x}\,.
 \label{sourcesum}
\eea
The latter are determined by the standard  Fourier expansion of the nonholomorphic Eisenstein series,
\be
E_s(x+iy) \ \  = \ \ \f{1}{2\zeta(2s)}\sum_{(c,d)\,\neq\,(0,0)}\f{y^s}{|c(x+iy)+d|^{2s}}    \ \  = \ \  \sum_{n\in\Z} \calF_{n,s}(y) \, e^{2\pi i n x}\,,
\label{eisenfourier}
\ee
where the zero mode consists of two power behaved terms,
\be
\calF_{0,s}(y)  \ \ = \ \  y^s \  +  \ \f{\sqrt \pi \,\Gamma(s-\half) \zeta(2s-1)}{\Gamma(s)\zeta(2s)}\, y^{1-s} \,,
\label{eisenzero}
   \ee
and the non-zero modes  are proportional to $K$-Bessel functions,
\be
\calF_{n,s}(y) \ \ = \ \  \f{2\,\pi^s}{\Gamma(s)\zeta(2s)}\,  |n|^{s-\half} \, \sigma_{1-2s}(|n|)
\sqrt{y}\,K_{s-\half}(2\pi |n|y) \,, \  \ \ n\neq 0
\label{nonzeroeisen}
\ee
(see \cite[\S1.6]{bump}).

Since the Laplace operator $\Delta_{\Omega}$ commutes with all group translations, the differential equation (\ref{laplaceeigen}) can be equivalently stated as the simultaneous differential equations
\begin{equation}
  \label{e:Diff}
  (y^2 \,\partial_y^2 - 12-4\pi^2n^2y^2) \, \widehat{f}_n(y) \ \  = \ \    S_n(y)\,, \  \ \ n \, \in\,\Z\,,
\end{equation}
for each Fourier mode of (\ref{fatstartofsect3}).
We will determine the solution for each value of $n$ in the form
\be
\widehat{f}_n(y)  \ \  = \ \  \widehat{f}^P_n(y)   \ + \  \widehat{f}^H_n(y) \,,
\label{e:fullform}
\ee
 where $\widehat{f}^P_n(y) $ is a particular solution to the equation
 and $ \widehat{f}^H_n(y)$ is a solution of the homogenous
 equation   which is chosen in order that  the solution $\widehat{f}_n(y)$ satisfies appropriate boundary conditions.

We now need to consider  these  boundary conditions.   The large $y$   (meaning weak string coupling)   growth condition  (\ref{largebound})  on $f(x+iy)$ carries over to each fourier coefficient $\widehat{f}_n(y)$, thus
\begin{equation}\label{boundaryconditionsforfn}
\widehat{f}_n(y)   \ \  = \ \    O(y^3) \ \ \text{for large~}\,y\,.
\end{equation}
In fact modes with $n\ne 0$ will be shown to decay
 like a constant times $y^{e_n} \exp(-2\pi |n| y)$ in this limit,
  with  values of $e_n$ that will be discussed later.

In addition to this boundary condition on each $\widehat{f}_n(y)$ for large $y$, there is also a condition for small $y$ which is in fact a consequence of (\ref{boundaryconditionsforfn}) together with the $SL(2,\Z)$-invariance of   $f(\Omega)$.  It is given by the following lemma.
\begin{lem}\label{lem:smally}
If $h(x+iy)$ is an $SL(2,\Z)$-invariant function on the upper half plane satisfying the large-$y$ growth condition $h(x+iy)=O(y^s)$ for some $s>1$, then each Fourier mode  $\widehat{h}_n(y)=\int_0^1 h(x+iy)e^{-2\pi i n x}dx$ of $h$ satisfies the bound $\widehat{h}_n(y)=O(y^{1-s})$ for small $y$.
In particular, the small-$y$ boundary condition for any mode number $n$ is
\be
\label{e:smally}
\widehat{f}_n(y) \ \  =  \ \ O(y^{-2})\,.
\ee
\end{lem}
\begin{proof}
Note the inequality  $E_s(x+iy)\ge y^s$ for $s>1$, which comes from dropping all terms with $c\neq 0$ in the definition (\ref{eisenfourier}).  By assumption, the large-$y$ bound states that there is a constant
 $C>0$ such that  $|h(x+iy) | \le Cy^s$ for any $x+iy$ in $\cF$, the standard fundamental domain for $SL(2,\Z)$.
It follows that
$|h(x+iy)|\le C E_s(x+iy)$ in $\cF$, and hence, by automorphy, everywhere in the upper-half plane.  This, together
with the fact that $E_s(x+iy) >0$, implies
\be
|\widehat{h}_n (y)| \ \  \le \ \  C\, \int_0^1 E_s(x+iy) \, dx \ \  = \ \  C\, \left(y^s +
\f{\sqrt \pi \,\Gamma(s-\half) \zeta(2s-1)}{\Gamma(s)\zeta(2s)}
  y^{1-s}\right).
\label{prooforigin}
\ee
Therefore $\widehat{h}_n (y) =O(y^{1-s})$
 as $y\to 0$.
 In the particular case $h=f$ and $s=3$,
the bound
 (\ref{largebound}) then implies (\ref{e:smally}).
\end{proof}

The conditions~\eqref{boundaryconditionsforfn} and~\eqref{e:smally} specify a unique solution to (\ref{e:Diff}).  To be explicit, we observe that
the solution space of the corresponding homogeneous differential equation
\begin{equation}\label{homogde}
  (y^2 \,\partial_y^2 - 12-4\pi^2n^2y^2) \widehat{f}^H_n \ \  = \ \   0
\end{equation}
consists of the two-dimensional space
\begin{equation}\label{homogeneoussolutions}
  \aligned
   \{ \widehat{f}^H_n & \ = \   a \sqrt{  y}\,I_{\sevenh}(2\pi |n|y) \,+\, b \sqrt{  y}\, K_{\sevenh}(2\pi |n|y) \,|\,a,b\in \C\} \, , \ \ n\,\neq\, 0\,, \\ \text{or~~~~~} \ \
\{ \widehat{f}^H_0  & \ = \  a y^4 \, + \, b y^{-3}
\,|\,a,b\in \C\}\,,  \ \ \ \ \qquad\qquad \ n\,=\,0\,,
\endaligned
\end{equation}
where the modified Bessel functions of the third kind are defined by
\begin{eqnarray}
   K_{7\over2}(y)&=& \sqrt{\smallf{\pi}{ 2y}}\, P(y)\, e^{-y}\\ \text{and} \ \ \  \
 I_{7\over2}(y)&=& \smallf{1}{\sqrt{2\pi y}}\, \left(P(-y)\,e^y+P(y)\, e^{-y}\right),
 \label{e:ikseven}
 \end{eqnarray}
 with $ P(y)= 1+\smallf 6y +\smallf{15}{y^2}+\smallf{15}{y^3}$.
The unique expression (\ref{e:fullform}) that satisfies the boundary conditions in the two dimensional solution space to~(\ref{e:Diff}), for $n\neq 0$, can be deduced by noting the following asymptotic behaviour of Bessel functions.  In the $y\to \infty$ limit the relevant functions behave as
\begin{equation}\label{asymptofIK7halves}
  \aligned
\sqrt{y}\,K_{7/2}(2\pi |n| y) &  \ \ = \ \  \f{e^{-2\pi |n| y}}{2|n|^{1/2}}\,(1+O(\smallf{1}{y}))\\
\text{and~~} \ \ \ \sqrt{y}\,I_{7/2}(2\pi |n| y) & \ \ = \ \  \f{e^{2\pi |n| y}}{2\pi|n|^{1/2}}\, (1+O(\smallf{1}{y^2}))\, ,
\endaligned
\end{equation}
so only the $K_{7/2}$ solution satisfies the boundary condition, which means that $a=0$ in~\eqref{homogeneoussolutions}.  The coefficient $b$ of the solution to the inhomogeneous equation is then determined by noting the $y\to 0$ asymptotics
\be
\label{smallyasym}
\sqrt{y}\,K_{7/2}(2\pi |n| y)  \ \ = \ \  \f{15}{16 |n|^{7\over2} \pi^3
  y^3} \ - \ {3\over 8|n|^{3\over2}\pi y} \ + \ O(y)
\ee
and imposing the condition  (\ref{e:smally}) for small $y$, which requires the $y^{-3}$ term in~\eqref{smallyasym} to cancel with a similar term in the particular solution $\widehat{f}^P_n(y)$. The situation for $n=0$ is of course simpler and again has $a=0$, and $b$ determined by asymptotics at the origin.

In order to analyze the particular solutions of~\eqref{e:Diff} we need first to discuss the
Fourier modes of the source term, which can be conveniently broken into a sum of products of Fourier modes  of the nonholomorphic Eisenstein series given in \eqref{eisenzero} and \eqref{nonzeroeisen},
\be
S_{n}(y)   \ \ = \ \   \sum_{\srel{n_1,n_2\,\in\,\Z}{n_1+n_2\,=\,n}} s_{n_1,n_2}(y)\,.
\label{componentsource}
\ee
The $s_{n_1,n_2}$ are naturally divided into the following classes:
\begin{itemize}
\item When $n_1=n_2=0$,
\be
s_{0,0}(y)  \ \ = \ \   -\,(2\,\zeta(3)\,y^{3\over2}\, +\, 4\,\zeta(2)\,y^{-{1\over2}})^2\,.
\label{e:Sn00}
\ee

\item
When either $n_1=0$ and $n_2=n\ne 0$ or $n_2=0$ and $n_1=n\ne 0$,
\begin{equation}
  \label{e:Sn0}
  s_{n,0}(y)  \ \  =  \ \    s_{0,n}(y)  \ \  =  \ \  -\,8\,\pi\,(2\,\zeta(3)\, y^2\, + \, 4\,\zeta(2)) \,
  \, {\sigma_{2}(|n|)\over |n|}\, K_1(2\pi |n|y)\,,
\end{equation}
where $\sigma_2(|n|)=\sum_{ k|n} k^2$, the sum being over positive   divisors.

\item    When $n_1 \ne 0$ and $n_2\ne 0$,
\begin{equation}\label{e:Sn1}
 s_{n_1,n_2}(y)  \ \  = \ \  -\,64\,\pi^2\,y\,
{\sigma_{2}(|n_1|)\,\sigma_{2}(|n_2|)\over |n_1n_2|} \, K_1(2\pi |n_1|y)\,K_1(2\pi |n_2|y)\,.
\end{equation}
\end{itemize}

In parallel with~(\ref{componentsource}), it will be useful to express  $\widehat{f}_n(y)$ as the sum
\begin{equation}\label{fn1n2}
  \widehat{f}_n(y) \ \   = \ \   \sum_{n_1+n_2\,=\,n}\widehat{f}_{n_1,n_2}(y)\,,
\end{equation}
where
\begin{equation}\label{fn1n2sn1n2}
   (y^2 \,\partial_y^2 - 12-4\pi^2(n_1+n_2)^2y^2) \, \widehat{f}_{n_1,n_2}(y) \ \  = \ \    s_{n_1,n_2}(y)\,.
\end{equation}
The space of solutions to this equation is again two dimensional and obviously shares the same homogeneous solutions given in (\ref{homogeneoussolutions}) with $n=n_1+n_2$.  There is an obvious ambiguity breaking apart (\ref{e:Diff}) into a sum of differential equations (\ref{fn1n2sn1n2}):~a homogeneous solution could be simultaneously added to one $\widehat{f}_{n_1,n_2}$ and subtracted from another $\widehat{f}_{n_1',n_2'}$, where $n_1'+n_2'=n_1+n_2$, without affecting the overall sum (\ref{fn1n2}).   To avoid this ambiguity, we shall insist that each $\widehat{f}_{n_1,n_2}$   satisfies  the same growth conditions as $\widehat{f}_n(y)$,
\begin{equation}\label{boundaryconditionsforfn1n2}
\aligned
\widehat{f}_{n_1,n_2}(y)  &  \ \ = \ \  O(y^3) \ \ \ \, \text{for $y$ large}\,,\\
\widehat{f}_{n_1,n_2}(y)  &  \ \ =  \ \   O(y^{-2}) \ \ \text{for $y$ small}\,.
\endaligned
\end{equation}
As before, such solutions are unique and   have the form
\begin{equation}\label{fn1n2sol}
    \widehat{f}_{n_1,n_2}(y)   \ \ = \ \   \widehat{f}^P_{n_1,n_2}(y) \ +
       \alpha_{n_1,n_2}\, \sqrt{y}\,K_{7/2}(2\pi|n_1+n_2|y)\,,
\end{equation}
for any values of $n_1$ and $n_2$,
where  $\widehat{f}^P_{n_1,n_2}$  is a particular solution satisfying the large-$y$ bound $O(y^3)$ and
$\a_{n_1,n_2}$ is the coefficient of the homogeneous solution, which will be determined by   the small-$y$ boundary condition $  \widehat f_{n_1,n_2}(y)=O(y^{-2})$.

We will now determine the explicit solutions for various choices of the integers $(n_1,n_2)$.
These give rise to the following sectors:

  \begin{enumerate}[(i)]
  \item $n_1=n_2=0$;
\item $n_1=0, n_2\ne0$  or  $n_1\ne 0, n_2 =0$;
\item  $n_1n_2>0$;
  \item $n_1n_2<0$;
\item $n_1,n_2\neq0$ and $n=n_1+n_2=0$.
  \end{enumerate}
  The last case is a special case of (iv) but merits separate discussion.
\subsection{Solutions of the equations in distinct sectors of $n_1$ and $n_2$}\hfill\break
\label{sec:solutions}

 {\bf (i) $n_1=n_2=0$}
\medskip

In this case the source term, $s_{0,0}(y)$, is given by the power behaved terms in~\eqref{e:Sn00} and it is easy to see that the solution to~\eqref{fn1n2sn1n2} is
\begin{equation}\label{f00}
  \widehat{f}_{0,0}^P(y)  \ \  = \ \   {2\,\zeta(3)^2\,\over3}\,y^3 \ + \ {4\,\zeta(2)\,\zeta(3)\over3}\,y \
  + \ {4\,\zeta(4)\over y}\,.
\end{equation}
Furthermore, $\alpha_{0,0}=0$ and  $\widehat{f}_{0,0}(y)=\widehat{f}^P_{0,0}(y)$.

The complete zero mode, $\widehat f_{0}(y)$ is given by the sum of  $\hat
f_{0,0}(y)$ and the terms of the form  $\widehat f_{n_1,-n_1}(y)$ that
arise in case~(v), and will be discussed  in section~\ref{completef}.

\bigskip
{\bf (ii) $n_1=0, n_2\ne0$} or  $n_1\ne 0, n_2 =0$
\label{onezero}
\medskip

It is easy to verify by substitution that~\eqref{fn1n2sn1n2} with source term~\eqref{e:Sn0}  has a  particular solution given by
\begin{multline}\label{exactfn0P}
 \widehat{f}_{n,0}^P(y)  \ \  = \ \   \widehat{f}_{0,n}^P(y)  \\  = \ \
   \frac{8 \,\sigma_2(|n|) }{9\,\pi\,|n|^3   }\times
\Big(q_{n,0}^0(\pi |n| y)
K_0(2 \pi |n | y)+ q_{n,0}^1(\pi |n| y)
K_1(2 \pi|n |  y)\Big),
\end{multline}
where the coefficients are given by
\begin{eqnarray}\label{coeffOP}
q_{n,0}^0(z)&=&{1\over z}\,
\left(90\zeta(3)-n^2\pi^4+9z^2\zeta(3)\right) \\
\text{and~} \ \ \
q_{n,0}^1(z)&=&{1\over z^2}\,\left(
90\zeta(3)-n^2\pi^4+54z^2\zeta(3)\right).
\end{eqnarray}
Note that the expression (\ref{exactfn0P}) respects the symmetries
 \begin{equation}\label{symmetriesinindices}
    \widehat{f}_{n_1,n_2}(y) \ \ = \ \ \widehat{f}_{-n_1,-n_2}(y) \ \ = \ \  \widehat{f}_{n_2,n_1}(y)\,.
 \end{equation}
Since  $ \widehat{f}_{n,0}^P(y) \sim -\f{4\sigma_2(|n|)(n^2\pi^4-90\zeta(3))}{9n^6\pi^4}\,y^{-3}$ as  $y\to 0$,
  the coefficient $\a_{n,0}$ of the second term in~(\ref{fn1n2sol}) must be
taken to be
\begin{equation}\label{exactalphan0}
    \a_{n,0}  \ \  =  \ \  \alpha_{0,n}  \ \  = \ \
    \f{64\,\sigma_2(|n|)\,(n^2\pi^4-90\,\zeta(3))}{135
\,|n|^{5\over2}\,\pi}
\end{equation}
in order that complete solution satisfies the boundary condition (\ref{boundaryconditionsforfn1n2}) at the origin.

Thus the  full solution~\eqref{fn1n2sol} given by
\begin{multline}
\label{fn0solution}
  \widehat{f}_{n,0}(y) \ \   = \ \ \widehat{f}_{0,n}(y) \ \   =  \\ \frac{8 \,\sigma_2(|n|) }{9\, \pi
    |n|^3}\times
\Big(q_{n,0}^0(\pi |n| y)
K_0(2 \pi |n | y)+ q_{n,0}^1(\pi |n| y)
K_1(2 \pi |n | y)\Big)\cr
+\f{64\,\sigma_2(|n|)\,(n^2\pi^4-90\,\zeta(3))}{135\,|n|^{5\over2}\,\pi} \,
\sqrt{y}\,K_{7\over2}(2\pi |n|y)
\end{multline}
 behaves as
\begin{equation}\label{asymptsoffn0Pat0}
   \widehat{f}_{n,0}(y) \ \   =  \ \  \widehat{f}_{0,n}(y) \ \  = \ \  \f{4\,\pi^2\,\sigma_2(|n|)}{15\,n^2\, y} \ + \   O(1)
\end{equation}
in the $y\rightarrow 0$ limit.  In the large-$y$ limit the solution behaves as
\begin{equation}\label{asymptsoffn0Patinf}
  \widehat{f}_{n,0}(y)  \ \  = \ \   \widehat{f}_{0,n}(y)  \ \ = \ \
    e^{-2\pi |n| y}\, \times \,\(4\, \sigma_2(|n|)  \,|n|^{-5/2}\,\zeta(3)\,y^{1\over2}   + O(1)\),
\end{equation}
where the exponential suppression has a form  characteristic of a
charge-$n$ D-instanton and the other factors are associated
with the instanton measure. This will be commented upon further in section~\ref{sec:stringy}.

\bigskip
{\bf (iii)  $n_1n_2> 0$}  and {\bf  (iv) }$n_1n_2<0$ with $n_1+n_2\ne 0$
\label{sec:samecond}
\medskip

 Let $\sgn(x)$ denote the sign function and  $H(x)=\f{1+\sgn(x)}{2}$ denote the Heavyside function.
It is easy to check   that an explicit
particular solution to~\eqref{fn1n2sn1n2} with source given
by~\eqref{e:Sn1} is given by the
 bilinear sum
  in $K_0$ and $K_1$ Bessel functions
\begin{multline}\label{exactfn1n2Psamesign}
\widehat{f}_{n_1,n_2}^{P}(y)  \ \ = \\
\f{32\,\pi\,\sigma_2(|n_1|)\,\sigma_2(|n_2|)}{3\,|n_1n_2|\,|n_1+n_2|^5} \sum_{ i,j=0, 1}
q^{ i,j}_{n_1,n_2}(\pi |n_1+n_2|y) \,  K_i(2 \pi |n_1| y)  \, K_j(2 \pi |n_2| y),
\end{multline}
where the matrix coefficients are given by the expressions
\begin{equation}\label{qsame00}
q^{ 0,0}_{n_1,n_2}(z)
=\sgn(n_1n_2)\,
(-4 z \, n_1n_2  \left(n_1 ^2+n_2^2 -6 n_1
   n_2 \right)-\smallf{30}{z}\, n_1 n_2  (n_1 - n_2)^2),
  \end{equation}
\begin{multline}\label{qsame01}
    q^{ 0,1}_{n_1,n_2}(z)   \ \ = \ \ \left(H(n_1n_2)\,+\,H(-n_1n_2)\,\sgn(n_1) \,\sgn(n_1+n_2)\right)\\
\times (
- n_1 \left(13 n_1^2  n_2 -65
   n_1  n_2^2+n_1 ^3+19 n_2^3\right)
+\smallf{30}{z^2} \,n_1  n_2^2 (n_1-n_2))\,,
\end{multline}
\begin{multline}\label{qsame10}
  q^{ 1,0}_{n_1,n_2}(z) \  \ = \ \  \left(H(n_1n_2)\,+\,H(-n_1n_2)\,\sgn(n_2) \,\sgn(n_1+n_2)\right)\\
\times
(- n_2 \left(13 n_2^2  n_1 -65
   n_2  n_1^2+n_2 ^3+19 n_1^3\right)
+\smallf{30}{z^2}\, n_2  n_1^2 (n_2-n_1))\,,
\end{multline}
and
\begin{multline}\label{qsame11}
q^{ 1,1}_{n_1,n_2}(z)  \  \ = \ \
-4 z\, n_1 n_2  \left(n_1^2+n_2^2-6 n_1
   n_2 \right)\\
- {14 n_1^3  n_2 -94  n_1^2 n_2^2+14 n_1  n_2^3+ n_1^4+ n_2^4\over z}\,.
\end{multline}
Imposing the small-$y$ boundary condition on $\widehat{f}_{n_1,n_2}(y)$  in  (\ref{fn1n2sol}) requires
\begin{multline}\label{alphasame}
 \alpha_{n_1,n_2}  \ \  = \ \, \sgn(n_1+n_2)\frac{128 \,\pi\,
    \sigma_2(|n_1|)\,\sigma_2(|n_2|)  }{45\, n_1^2 \, n_2^2\,
    |n_1+n_2|^{7\over2}}  \,
   \Bigl(n_1^5 +n_2^5+15 n_1^4 n_2 +15 n_1   n_2^4\cr
-80 n_1^3
   n_2^2 -80 n_1^2 n_2^3
+60 n_1^2 n_2^2
   (n_1-n_2) \log (|\smallf{n_1}{n_2}|)\Bigr),
\end{multline}
and the resulting $y\to 0$ behaviour of (\ref{fn1n2sol})  is given by
\begin{equation}\label{fn1n2samesignasymp0}
\widehat{f}_{n_1,n_2}(y)   \ \ = \ \    \frac{8\, \sigma_2(|n_1|)\,\sigma_2(|n_2|)}{5 \,n_1^2 \,  n_2^2\,
   y} \ + \ O(1)\,.
\end{equation}

In  sector~(iii), where
$|n_1+n_2|=|n_1|+|n_2|$, the
$y\to \infty$ behaviour  of~(\ref{fn1n2sol})    has the instantonic form
\begin{multline}\label{fn1n2samesignasympinf}
\widehat{f}_{n_1,n_2}(y)  \ \ = \ \      e^{-2\pi|n_1+n_2|y}\,\Big({ \alpha_{n_1,n_2}\over2|n_1+n_2|^{\frac12}}\\ -\f{
 64 \, \pi^2 \,   \sigma_2(|n_1|)\, \sigma_2(|n_2|)}{3 \, |n_1n_2|^{\frac12} }\,{n_1^2+n_2^2-6n_1n_2\over(n_1+n_2)^4}
+O(y^{-1}) \Big)\, .
\end{multline}
In  sector~(iv), where $|n|= |n_1+n_2| <|n_1|+|n_2|$ a qualitatively new feature is that there are an infinite number of values of $n_1$ and $n_2$ having
 a fixed value of  $n=n_1+n_2$.
 Because of this,  the $y\rightarrow\infty$ limit is very different from the large-$y$ limit for  the  $n_1 n_2>0$ case in~\eqref{fn1n2samesignasympinf}
 since the particular solution contains terms that decrease exponentially relative to BPS D-instanton terms.  Explicitly, when $n_1 n_2 <0$ the large-$y$ behaviour is  given by
\begin{equation}\label{fn1n2oppsignasympinf}
\aligned
 \widehat{f}_{n_1,n_2}(y)  \ \  = \ \ &   \alpha_{n_1,n_2}\,\sqrt{y}\,K_{7/2}(2\pi|n_1+n_2|y) \\
&  \ \ \ \ - \   e^{-2\pi(|n_1|+|n_2|)y}\left(\f{
  \sigma_2(|n_1|)\,\sigma_2(|n_2|)}{  |n_1 n_2|^{\frac52}\,y^2 }\  +\ O(y^{-3})\right)
 \\
 = \ \   &   e^{-2\pi|n_1+n_2|y} \, { \alpha_{n_1,n_2}\over2|n_1+n_2|^{\frac12}}(1+O(y^{-1})) \\
&  \ \ \ \ - \   e^{-2\pi(|n_1|+|n_2|)y}\left(\f{
  \sigma_2(|n_1|)\,\sigma_2(|n_2|)}{  |n_1 n_2|^{\frac52}\,y^2 } \ + \ O(y^{-3} )\right).\\
\endaligned
\end{equation}
The second term in   either expression  can be more exponentially damped than the first term as  $n_1$ or
$n_2$ increases with  $n=n_1+n_2$ held fixed.

\bigskip
{\bf  (v) $n_1,n_2\ne 0$} with  $n=n_1+n_2=0$
\medskip

This is a special case of (iv) and the particular solution can now be
obtained by carefully considering the limit $n_2=- n_1 + \epsilon$
with $\epsilon\to 0$ in~\eqref{exactfn1n2Psamesign}.  Superficially,
the presence of the $|n_1+n_2|^{-5}$ factor there suggests that this
limit gives a badly divergent result.  However, there are massive
cancelations   caused by properties of the $K$-Bessel functions and the
resulting limit simplifies to be
\begin{equation}\label{n1isnegn2}
\widehat{f}_{n_1,-n_1}^{P}(y)   \ \  =   \ \  \frac{32\,\pi\,
  \sigma_2(|n_1|)^2}{315\,|n_1|^3}
\sum_{i,j=0,1}  r^{i,j}(\pi |n_1| y) \, K_i(2\pi |n_1|y)\, K_j(2\pi|n_1|y)
\end{equation}
where the coefficient matrix, $r^{ij}$, has components
\begin{eqnarray}
r^{0,0}(z)&=& z \left(-512 z^4+48 z^2-15\right) \nonumber \\
r^{0,1}(z)&=&r^{1,0}(z) \ \ = \ \  -   \left(128 z^4+12 z^2+15\right) \label{n1isnegn2det} \\
\nonumber r^{1,1}(z)&=&z^{-1}\left(512 z^6+16 z^4+33z^2-15\right).
\end{eqnarray}
The solution of the homogeneous equation solution can also be obtained by setting $n_2 = -n_1+\epsilon$ in $\alpha_{n_1,n_2}$ and considering the limit $\epsilon\to 0$, which leads to
\be
\label{limzeroone}
    \lim_{n_2\to -n_1}   \alpha_{n_1,n_2}\,
\sqrt{y}\,K_{7/2}(2\pi|n_1+n_2|y)  \ \ = \ \  \frac{8\, \sigma_2(|n_1|)^2}{21\,
 n_1^6\, \pi^2\, y^3}\,.
    \ee
In order to verify that the full solution
\begin{equation}\label{n1isnegn2sol}
     \widehat{f}_{n_1,-n_1} (y)  \ \ = \ \  \widehat{f}_{n_1,-n_1}^{P}(y)
\  + \
     {8\,\sigma_2(|n_1|)^2\over 21 \,n_1^6 \,\pi^2\, y^3}\,
\end{equation}
 satisfies the $y=0$ boundary condition, we also note that for small $y$
\be
\label{limzerotwo}
 \widehat{f}_{n_1,-n_1}^{P}(y)  \ \     = \ \
-\frac{8\,\sigma_2(|n_1|)^2}{21\, n_1^6\, \pi^2 \, y^3} \ + \
\frac{8\,\sigma_2(|n_1|)^2}{5\, n_1^4\,  y}  \ +\ O(1)\,.
\ee
Therefore, it  follows from~\eqref{n1isnegn2sol} that at small $y$ the full solution for the $(n_1,-n_1)$ mode is
\begin{equation}\label{n1isnegn2asymptatinf}
 \widehat{f}_{n_1,-n_1} (y) \ \   =  \ \   \frac{8 \,\sigma_2(|n_1|)^2 }{5\, |n_1|^4  \, y} \ + \ O(1)\,,
\end{equation}
and  at large  $y$ it is
\begin{equation}\label{n1isnegn2asymptatinf}
   \widehat{f}_{n_1,-n_1} (y)  \ \  = \ \   \f{8\,\sigma_2(|n_1|)^2}{21\,
     n_1^6\, \pi^2\, y^3}  \  - \  e^{-4\pi |n_1| y}\Bigl(\frac{\sigma_2(|n_1|)^2}{|n_1|^5\, y^2 } + O(y^{-3})\Bigr).
\end{equation}
Note that the power behaved term proportional to $1/y^3$ was uncovered by a different method in~\cite{Green:2005ba}  and is interpreted as a genus-three contribution to the amplitude in string theory perturbation theory.  The exponentially decaying term is characteristic of the contribution of a charge-$(n_1,-n_1)$ D-instanton/anti D-instanton pair.

\subsection{The complete expression for each Fourier mode, $\widehat{f}_n(y)$}\hfill\break
\label{completef}

Having determined the expressions for
 $\widehat{f}_{n_1,n_2}(y)$ we
shall now study the $n$-th mode $\widehat{f}_n(y)$, which we recall was
 given in~\eqref{fn1n2} as  the sum of   $\widehat{f}_{n_1,n_2}(y)$
 over $n_1$ and $n_2$ with  $n_1+n_2 = n$.   We
 first note that by (\ref{f00}) and the explicit formulas for each
 $\widehat{f}_{n_1,n_2}$ given in section~\ref{sec:solutions}, the
 $SL(2,\Z)$-invariant function
 $f(\Omega)-\f{2\zeta(3)^2}{3}E_3(\Omega)$ is $O(y)$ for $y$ large.
 Applying lemma~\ref{lem:smally}, we conclude that its Fourier
 coefficients $\widehat{f}_n(y)-\f{2\zeta(3)^2}{3}{\mathcal
   F}_{n,3}(y)$ obey the bound $O(y^{-\e})$ for any fixed positive
 real number $\e>0$.   Using formulas
   \eqref{eisenzero} and~\eqref{nonzeroeisen}, this  gives the asymptotic statement
 \be
  \widehat{f}_n(y)  \ \ =  \ \ {945\,
  \zeta(3)^2\,\sigma_{-5}(|n|)\over 4\,\pi^5} {1\over y^2} \  + \ O(y^{-\e})\,,
\label{nonzerolimit}
\ee
 again for any fixed $\e>0$.  In the case $n=0$, $\sigma_{-5}(|n|)$ should be interpreted as $\zeta(5)$.  The error term can be slightly improved using the Kronecker limit formula, though this will not be important for our purposes.
Note that even though each term in  \eqref{fn1n2}  satisfies the small-$y$ bound $O(y^{-1})$, their aggregate sum diverges like $y^{-2}$ in (\ref{nonzerolimit}).

\medskip
\noindent{\bf The constant term:} The $n=0$ mode is given by
 \be
\widehat{f}_0(y) \ \  = \ \
 \widehat{f}_{0,0}(y) \  +  \ \sum_{n_1\ne 0} \widehat{f}_{n_1,-n_1}(y)\,.
\label{zerodef}
\ee
The sum of the second term in ~\eqref{n1isnegn2sol} over all nonzero integers $n_1$ is
\begin{equation}\label{gettinginversecubed1}
  \f{16}{21\,\pi^2\, y^3}  \sum_{m>0}{\sigma_2(m)^2\over m^6}
\ \   = \ \   \f{16}{21\,\pi^2\, y^3}\,  \f{\zeta(6)\,\zeta(4)^2\,\zeta(2)}{\zeta(8)}
\ \  = \ \  \f{4\,\zeta(6)}{27\, y^3}\,,
\end{equation}
where we have used the Ramanujan identity
\begin{equation}
  \sum_{m\,=\,1}^\infty    {\sigma_{t}(m)\,\sigma_{t'}(m)\over   m^r} \ \ = \ \
  {\zeta(r)\,\zeta(r-t)\,\zeta(r-t')\,\zeta(r-t-t')\over \zeta(2r-t-t')}\,.
  \label{eq:ramanujanid}
\end{equation}
As a result of this and (\ref{f00}), we can write the complete solution for the zero mode  as
\be
\widehat{f}_0(y) \ \  =  \ \  {2\,\zeta(3)^2\over3}\,y^3\ +
 \ {4\,\zeta(2)\,\zeta(3)\over3}\,y \ + \
{4\,\zeta(4)\over y}  \ + \   \f{4\,\zeta(6)}{27\, y^3} \  + \   \sum_{n\ne 0} \widehat{f}^{P}_{n,-n}(y)\,,
\label{zeromode}
\ee
where the expression for $\widehat{f}^P_{n,-n}(y)$ is given
in~\eqref{n1isnegn2} and is exponentially suppressed as $y\to \infty$.
The behaviour as $y\to 0$ is more subtle since the sum in
\eqref{fn1n2} does not commute with the
small-$y$ limit, and was given above in (\ref{nonzerolimit}).
 A finer asymptotic expansion can be obtained using Mellin transform methods.

\medskip
\noindent{\bf The non-zero Fourier modes:}
Modes with $n \ne 0$ get contributions from the sectors labelled (ii), (iii) and (iv),  so that,
 \begin{equation}
    \label{e:modesn}
 \widehat{f}_n(y)  \ \ = \ \  \widehat{f}_{n,0}(y) \  + \
 \widehat{f}_{0,n}(y)  \ + \
 \sum_{n_1\,=\,1}^{n-1} \widehat{f}_{n_1,n-n_1}(y) \  + \  2 \sum_{n_1\geq
   n+1} \widehat{f}_{n_1,n-n_1}(y)\,.
 \end{equation}
 It is first  important to verify that the last sum is  convergent.   This
 involves an estimate of the behaviour of its terms as $|n_1|
 \to \infty$, which arises  in case (iv).
 The $K_i(2\pi |n_1|y) K_j(2\pi |n-n_1|y)$ terms in the $n_1$ sum (coming from (\ref{exactfn1n2Psamesign}) are
 exponentially suppressed as $|n_1|$ gets large.  Furthermore, for fixed $n$, an analysis of formula (\ref{alphasame}) shows that $\alpha_{n_1,n-n_1}=O(n_1^{-6})$ as $n_1\rightarrow\infty$.  Thus the terms coming from the homogeneous solutions $\alpha_{n_1,n_2}\sqrt{y}K_{7/2}(2\pi|n_1+n_2|y)$  also converge because the sum
 $\sum_{n_1=-\infty}^\infty \alpha_{n_1,n-n_1}$  is finite.

The leading behaviour in the weak coupling limit  $y\to\infty$   has the form
\begin{eqnarray}
    \widehat{f}_n(y)  \ \ = \ \  e^{-2\pi |n| y}\(8\,\f{\sigma_2(|n|)}{|n|^{5/2}}\,\zeta(3)\,y^{1/2}   + O(1)\),
\end{eqnarray}
which is dominated by the
behaviour of  $\widehat{f}_{n,0}$ and  $\widehat{f}_{0,n}$.  The behaviour for small $y$ was given in (\ref{nonzerolimit}).   It is also possible to study these asymptotics using the explicit formulas for $\widehat{f}_{n_1,n_2}$ given in section \ref{sec:solutions}, or from an analysis of (\ref{partialsum}) (which gives an alternative description of the terms in~\eqref{e:modesn}).  See also formula (\ref{kloosterfourier}), which gives yet another formula for $\widehat{f}_n(y)$.


\section{Discussion and  connections with string theory}
\label{sec:stringy}

The motivation for considering the differential equation ~\eqref{laplaceeigen} from~\cite{Green:2005ba}  was based on considering the compactification of the two-loop Feynman diagrams
of  the four-graviton amplitude of eleven-dimensional supergravity on $T^2$, in the zero-volume limit, $\cV
\to 0$.  The first non-leading term in the
low-energy expansion  of this amplitude was argued in~\cite{Green:2005ba} to give  the effective type~IIB string theory  interaction $f(\Omega)\, D^6R^4$,
with $f(\Omega)$ satisfying~\eqref{laplaceeigen}.  In this paper we have determined the exact solution for all the Fourier modes  $\widehat f_n(y)$ from (\ref{fouriermodes}).

The zero mode $\widehat f_{0}(y)$ \eqref{zeromode} possesses four
terms that are power behaved in $ y$ that were originally discussed in
detail in~\cite{Green:2005ba}.  The coefficients of these powers are
rational numbers multiplying products of zeta values.  The values of
these coefficients should agree with explicit perturbative string
theory calculations up to genus three.  The genus zero and genus one
string results were known to agree at the time of publication
of~\cite{Green:2005ba}.  The genus-two contribution has been
related in~\cite{D'Hoker:2013eea} to the integral of an invariant
introduced in~\cite{Zhang,Kawazumi}, which has also recently been
evaluated \cite{D'Hoker:2014} and agrees with the genus-two term
(the $y^{-1}$ term in \eqref{zeromode}).  The genus-three part (the
$y^{-3}$ contribution in~\eqref{zeromode}) agrees precisely with the
prediction for that term in the type IIA theory, that arises from the
expansion of the one-loop eleven-dimensional supergravity amplitude
compactified on a circle~\cite{Green:1999pu}.  Although a recent
genus-three string theory calculation~\cite{Gomez:2013sla} also apparently
reproduces this value for the coefficient of the  $y^{-3}$ contribution, there are
currently some questions concerning technical details of the calculation.

In solving for the modes $\widehat f_n(y)$, it was important to
understand the nature of the boundary conditions at $y=\infty$ and
$y=0$.  Although the condition at large $y$ (the weak coupling regime)
is simply that no term can be more singular than $y^3$, which is the
power corresponding to tree-level perturbation theory, the condition
at $y=0$ is more subtle.  We showed in lemma~\ref{lem:smally} that the
necessary condition is that $ \widehat f_n(y) = O(y^{-2})$ in the
limit $y\to 0$, which follows as consequence of $SL(2,\Z)$ invariance
together with the $y\to \infty$ bound, $\widehat{f}_n(y) =O(y^3)$.
This is a highly non-trivial condition, in that it implies that the
infinite series of terms that manifests itself as a series of
exponentially decreasing D-instanton and anti D-instanton
contributions at large $y$, simultaneously conspires to cancel a
singular term in $\widehat f_n(y)$ at small $y$.  This bears some
similarity to the behaviour of the metric on the Coulomb branch of
three-dimensional ${\mathcal N}=4$ supersymmetric $SU(2)$ Yang--Mills
theory with no flavour fields in Seiberg-Witten
theory~\cite{Seiberg:1996nz} (see also \cite{Dorey:1997ij}). In that
case, the expansion of the moduli space metric at large values of the
Higgs field also gets contributions from an infinite series of
exponentially suppressed terms~\cite{Gibbons:1995yw}, but the solution
can be uniquely determined by requiring the Coulomb branch metric to
be non-singular at the origin\footnote{It has been suggested that the
  series of exponentially suppressed terms might be interpreted as
  instanton/anti-instanton
  contributions~\cite{Hanany:2000fw}. However, the identification of
  the radial coordinate in the Atiyah-Hitchin metric with the
  corresponding scalar vacuum expectation value in the explicit
  semi-classical solution is ambiguous.  Owing to the high degree of
  supersymmetry in our case, it is not possible to redefine the
  modular parameter $\Omega$ without losing $SL(2,\Z)$ invariance, so
  this ambiguity is not present.}.

The expressions for the Fourier modes   contain detailed information concerning the instanton-like  contributions that decrease exponentially at large $y$.   Such  terms that have the form expected of contributions arising from D-instantons,  anti D-instantons and D-instanton/anti D-instanton pairs.    This is explicit in the large-$y$ limits given in~\eqref{n1isnegn2}  for the terms contributing to $\widehat f_0(y)$ and in~\eqref{asymptsoffn0Patinf},~\eqref{fn1n2samesignasympinf} and~\eqref{fn1n2oppsignasympinf} for the terms contributing to $\widehat f_n(y)$.   In particular,~\eqref{n1isnegn2asymptatinf} shows that the constant term, $\widehat f_0(y)$,  has an infinite series of exponentially decreasing terms in the large-$y$ limit, which have exponential factors $e^{-4\pi |n|y}$ that have the form which would arise from a D-instanton/anti D-instanton pair with charges $n$ and $-n$. Furthermore, the measure contains the square of the divisor sum $\sigma_2(|n|)$,
\be
\label{zeromeasure}
e^{-4\pi |n|y}  \,\f{\sigma_2(|n|)^2}{|n|^5}\,{1\over y^2} \,.
\ee
Since the measure for a single charge-$n$ D-instanton contains a single power of a divisor sum, this is another indication that terms of this form in $\widehat f_0(y)$ might be identified with  D-instanton/anti D-instanton pairs. Such instanton/anti-instanton terms should break all supersymmetries, giving rise to extra fermionic zero modes.  Soaking these up should ought to account for the fact that they are suppressed by the factor of $1/y^2$ in~\eqref{zeromeasure}, although we have not determined such factors  in the measure from an explicit D-instanton calculation.

The exponentially suppressed terms that  contribute to $\widehat f_{n,0}$ and $\widehat f_{0,n}$ with $n\ne 0$ might be interpreted as contributions of  single charge-$n$ D-instantons or charge-$n$ anti D-instantons with a measure that can be read off from~\eqref{asymptsoffn0Patinf},
\be
\label{halfbps}
e^{-2\pi |n| y}\left(\f{4\sigma_2(|n|)}{|n|^{5/2}}\,\zeta(3)\,y^{1/2}  + O(1)\right) ,
\ee
which has a factor of $y^{5/2}$ relative to~\eqref{zeromeasure}.  Likewise, the large-$y$ contribution to $\widehat
f_n(y)$ with $n=n_1+n_2$ and   ${\rm sign}(n_1)={\rm sign} (n_2)$, obtained in~\eqref{fn1n2samesignasympinf} has the form
\be
 e^{-2\pi|n_1+n_2|y}\, \sigma_2(|n_1|)\,\sigma_2(|n_2|)
 \times ({\rm function\ of}\ n_1, n_2) \,,
\label{interfn}
\ee
which has a power of $y^0$.

It would be desirable to understand the particular powers of $y$ in the prefactors of \eqref{zeromeasure}, \eqref{halfbps} and \eqref{interfn} in terms of the zero modes associated with supersymmetry breaking,  but we have not understood this in a systematic manner.

In any case, given the  non-standard application of M-theory/string theory duality that motivated \eqref{laplaceeigen}, we would like to determine whether this equation accurately describes the coefficient of the $D^6R^4$ interaction beyond the checks outlined above.
Further motivation for this equation and its generalisation to higher-rank duality groups was obtained in~\cite{Green:2010wi,Green:2010sp,Basu:2007ck} in considering properties of  the low energy effective action of type II string theory in lower dimensions obtained by toroidal  compactification to dimension $D$.  In these cases the coefficient of the $D$-dimensional $D^6 R^4$ interaction, $f^{(D)}$, is a function of the moduli associated with the  $E_{11-D}(\Z)$ duality group\footnote{Recall that the duality groups of rank $\le 8$   are specific real split forms of  $SL(2,\Z)$, $SL(3,\Z)\times SL(2,\Z)$, $SL(5,\Z)$, $Spin(5,5,\Z)$, $E_6(\Z)$, $E_7(\Z)$, $E_8(\Z)$.}.   Equation
 \eqref{laplaceeigen}  then  generalises to an inhomogeneous Laplace eigenvalue equation \cite{GMRV}
\begin{equation}
\left( \Delta^{(D)} -{6(14-D) (D-6)\over D-2} \right)\,f^{(D)}=-\left(\cE_\threeh^{(D)}\right)^2 + 120\,\zeta(3)\,  \delta_{D-6,0}\,,
\label{laplaceeigenthree}
\end{equation}
where $\Delta^{(D)}$is the laplacian on the homogeneous space and $\cE_s^{(D)}$ is the maximal parabolic Langlands Eisenstein series attached to the parabolic associated with the first node of the Dynkin diagram
(which is the coefficient of the $R^4$ interaction in $D$ dimensions).
The constant terms in various parabolic subgroups were analysed  to a
certain extent for the cases with $D\ge 6$  in \cite{Green:2010wi,
  Green:2010sp}  and for  $D=3$ in \cite{GMRV}, and agreed with
expectations based on perturbative string theory calculations.  This
has also been extended to the cases of
  $D=1$ and 2 in \cite{Fleig:2012xa}.  The analysis of the non-zero Fourier modes  presents new challenges that extends the considerations of \cite{GMV}, which considered the maximal parabolic Langlands Eisenstein series that arise as coefficients of the $R^4$ and $D^4R^4$ interactions.
The four dimensional version of \eqref{laplaceeigenthree} has also
received support from consideration of  the
soft scalar limits of $\mathcal N=8$ supergravity amplitudes in four dimensions~\cite{Beisert:2010jx} .

 Since the natural region of validity of perturbative supergravity is
 ${\mathcal V}\gg \ell_{11}^2$ it is not obvious why the M-theory
 argument that leads to $f(\Omega)$  should be a good approximation to
 the exact answer.  However, in common with analogous duality
 arguments for BPS quantities, the fact that the $D^6R^4$ interaction
 is $\eighth$-BPS seems to justify what would  otherwise be an
 outrageous continuation in ${\mathcal V}$.   In considering higher
 order interactions in the low energy expansion there is no reason,
 based on our current understanding, for expecting  such a
 continuation from large to small ${\mathcal V}$ to be valid.
  Nevertheless, it might be  of interest to analyze the structure
  of the compactified Feynman diagrams of eleven-dimensional
  supergravity further, if only to find inspiration for the possible
  mathematical structure of higher order terms.
The paper \cite{Green:2008bf} contains a detailed discussion of  higher order corrections to the low energy
expansion, that arise by expanding the two-loop four-graviton amplitude
of eleven-dimensional supergravity to higher orders beyond the $D^6R^4$ interaction studied in this paper.   This does not yield any contributions that survive
the $r_B\to \infty$ limit  to $D=10$ dimensions, but  does give
contributions that may be useful at finite values of  $r_B$ (i.e., in
the $D=9$ type~IIB theory).  Even though the analysis
in~\cite{Green:2008bf} is not the complete story,
the equations that emerge from the higher order expansion of the
two-loop amplitude   suggest that~\eqref{laplaceeigen} is
a specially simple example of a more general set of equations for the higher-order coefficients.

\bigskip
{\bf Acknowledgments}

We are particularly grateful for discussions with Don Zagier during
early stages of this work, and to Axel Kleinschmidt and Boris Pioline for their helpful comments.  We are also grateful to Nick Dorey, Dorian
Goldfeld, Savdeep Sethi, and David Tong  for clarifying discussions.

The research leading to these results has received funding from the
European Research Council under the European Community's Seventh
Framework Programme (FP7/2007-2013) / ERC grant agreement
no. [247252],  and the ANR grant   reference QST ANR 12 BS05 003
01, and the PICS  6076.  SDM was supported by NSF grant DMS-1201362.
We are also grateful to the Simons Center, Stony Brook where part of this research was undertaken in October 2013.   MBG is also grateful to the Aspen Center for Physics for its hospitality  in the Summer of 2013.

\appendix

\section*{Appendices:~Other methods to solve the inhomogeneous Laplace equation~(\ref{laplaceeigen})}

In the following three appendices we will briefly describe other approaches to constructing solutions to (\ref{laplaceeigen}). None of them is completely satisfactory, though each reveals different information about the solution; further variants of these might be useful in studying the $D^6R^4$ coefficient in lower dimensions (equivalently, for higher rank groups).

\section{The solution as a Poincar\'e Series}
\label{sec:poincaremeth}

The procedure used to determine the solutions for
  $\widehat{f}_n(y)$  in section~\ref{sec:fouriertransform} has the drawback that
  the modularity properties of the complete solution $f(\Omega)$ in (\ref{fouriermodes})
   are obscured in the mode-by-mode analysis.  In particular,  the values of the coefficients
     $\alpha_{n_1,n_2}$ in \eqref{fn1n2sol} were determined only by invoking a boundary
     condition at $y=0$.  This appendix presents a solution for $f(\Omega)$ in terms of
      a Poincar\'e series  that makes both the $SL(2,\Z)$-invariance  and growth conditions
      of each $\widehat{f}_{n_1,n_2}(y)$ manifest.  The results of this section should be viewed as complementary to those of section~\ref{sec:fouriertransform}, where the modes were expressed very explicitly in terms of products of $K$-Bessel functions (whereas here they will be given in terms of integrals).

\subsection{Sum over translates of one-dimensional solution}

Squaring the definition of $E_{3/2}$  given in~\eqref{eisenfourier}, we write
\begin{equation}\label{E32squared}
\aligned
  E_{3\over2}(\Omega)^2 \ \ &  = \ \ \sum_{\g_1,\g_2\,\in\,\G_\infty\backslash \G}\im\!(\g_1\Omega)^{3/2}\,\im\!(\g_2\Omega)^{3/2}
   \\ & = \ \ E_3(\Omega) \ + \  \sum_{\g_1\neq \g_2\,\in\,\G_\infty\backslash \G}\im\!(\g_1\Omega)^{3/2}\,\im\!(\g_2\Omega)^{3/2} \\
   & = \ \ E_3(\Omega) \ + \  \sum_{\g\,\in\,{\mathcal S}} (\det \g)^{-3}\, \calT(\g\Omega)\,,
  \endaligned
\end{equation}
where $\G=SL(2,\Z)$,  $\G_\infty =\{\ttwo{\pm 1}{\star}{0}{\pm 1} \in SL(2,\Z)\}$,
${\mathcal S}=\{\pm 1\}\backslash \{\ttwo{m_1}{n_1}{m_2}{n_2}\in M_2(\Z)\cap GL^+(2,\IR)|\gcd(m_1,n_1)=\gcd(m_2,n_2)=1\}$, and
\begin{equation}\label{poincare1}
\calT(\Omega)  \ \ =  \ \   \calT(x+iy) \ \ :=  \ \ \sigma(\smallf xy)\,, \ \ \sigma(u)   \ = \   (u^2+1)^{-3/2}.
\end{equation}
Indeed, this can easily be seen using the calculations
\begin{equation}\label{poincare7aa}
    \f{\Re\g\Omega}{\im\g\Omega} \ \ = \ \ \f{n_1n_2+m_2n_1 x + m_1n_2 x + m_1m_2(x^2+y^2)}{y\,\det\g}
\end{equation}
and
\begin{multline}\label{poincare7b}
    \((n_1n_2+m_2n_1 x + m_1n_2 x + m_1m_2(x^2+y^2))^2+(y\det \g)^2\)^{-3/2} \\ = \ \
    \(|m_1z+n_1|^2|m_2z+n_2|^2\)^{-3/2}.
\end{multline}
We shall use the fact that $\mathcal T(x+iy)$ only depends on the single parameter $\f xy$ to construct a solution to (\ref{laplaceeigen}) from (\ref{E32squared}).

Consider the differential equation
\begin{equation}\label{poincare4}
  \left(\f{d}{d u}\, \left((1+u^2)\,\f{d}{d u} \right)-12\right)h(u)  \ \  = \ \   -\,\sigma(u)  \,,
\end{equation}
 where the differential operator on the lefthand side corresponds to $\Delta_\Omega-12$ acting on functions of the ratio $u=\f xy$.

\begin{lem}\label{inhomtosigma}
The function\footnote{ The first term of the
  solution~\eqref{poincare3}  (which solves (\ref{poincare4}))    was essentially obtained in~\cite{Green:2005ba},
   but in that
   reference  the second term  (which solves the homogeneous version of (\ref{poincare4}))
   was instead a function of  the two variables  $x$ and $y$ rather than only of their
   ratio. The combination of terms in~(\ref{poincare3}), which was pointed out to us by
Don Zagier, has asymptotic behaviour that guarantees convergence of the sums that arise in (\ref{poincare7old}), whereas the expression used in \cite{Green:2005ba}  leads to a divergent result.}
\bea
h(u) \   =   \ {7+44u^2+40u^4\over3\sqrt{1 + u^2}} \, - \, \frac{16}{3\pi}\left(\frac{4}{3} + 5 u^2 + u (3 + 5 u^2) \tan^{-1}(u)\right)
\label{poincare3}
\eea
is the unique smooth, even function satisfying both (\ref{poincare4}) and the
 decay condition $h(u)\sim \f{1}{6|u|^{3}}$ as $u\rightarrow\pm\infty$.  It furthermore
  extends to a holomorphic function on $\C-\{iv\,|\,|v|\ge 1\}$, with jump discontinuities along
  these branch cuts given by
\begin{multline}\label{branchposneg}
  \lim_{u\rightarrow 0^+}h(u+iv)  \ - \   \lim_{u\rightarrow 0^-}h(u+iv) \ \ = \\ \left\{
                                                                               \begin{array}{ll}
                                                                                 {2 i\over3} \left(8 v \left(5 v^2-3\right)-{40 v^4-44 v^2+7\over\sqrt{v^2-1}}\right), & v > 1 \\
                                                                                { 2 i\over3} \left(8 v \left(5 v^2-3\right)+{40 v^4-44 v^2+7\over\sqrt{v^2-1}}\right), & v<-1.
                                                                               \end{array}
                                                                             \right.
\end{multline}
\end{lem}

\noindent
With $h(u)$ as defined in (\ref{poincare3}), define
 \begin{equation}\label{poincare6}
  F(\Omega) \ \  = \ \  F(x+iy) \ \  = \ \  4\,\zeta(3)^2\,h(\smallf xy)
\end{equation}
so that  $(\Delta-12)F= -4\zeta(3)^2 {\mathcal T}$.
Appealing to the expression (\ref{E32squared}),   the general solution to~\eqref{laplaceeigen} among automorphic functions having polynomial growth  has the form
\begin{equation}\label{poincare7old}
f(\Omega)   \ \ =  \ \ \smallf{2\,\zeta(3)^2}{3}\,E_3(\Omega)  \  + \  \sum_{\g\,\in\,\mathcal S}
  (\det \g)^{-3} F(\g \Omega)  \ + \ \a\,E_4(\Omega)\,, \ \ \a\,\in\,\C\,.
\end{equation}
 Note that $h$ is bounded by a constant multiple of $\sigma$, and so $F$ is bounded by a constant multiple of $\mathcal T$.
 Thus the absolute convergence of the sum (\ref{poincare7old}) follows from that of (\ref{E32squared}).  This argument  furthermore   shows that the $\g$-sum is $O(y^3)$ as $y\rightarrow \infty$;
since $E_3(\Omega)$ satisfies the same bound, this implies the coefficient
\begin{equation}\label{alphavanishes}
    \alpha \ \ = \ \ 0
\end{equation}
in order for (\ref{largebound}) to hold.  Thus (\ref{poincare7old}) simplifies to
\begin{equation}\label{poincare7new}
f(\Omega)   \ \ =  \ \ \smallf{2\,\zeta(3)^2}{3}\,E_3(\Omega)  \  + \  \sum_{\g\,\in\,\mathcal S}
  (\det \g)^{-3} F(\g \Omega)
\end{equation}
as a result.

Consider the function
\begin{equation}\label{asympsofphi1}
      \phi\(\ttwo{m_1}{n_1}{m_2}{n_2},\Omega\) \ \ := \ \ |m_1n_2-n_1m_2|^{-3}F(\smallf{m_1\Omega+n_1}{m_2\Omega+n_2})\,.
\end{equation}
For  $\Omega$  fixed, it
has the well-defined limit  ${2\zeta(3)^2\over3}\, \f{ y^3}{|m_2\Omega+n_2|^6}$ along the singular set $m_1n_2-n_1m_2=0$, which we shall take as its defining value there.  Thus  (\ref{poincare7new}) can be rewritten as
\begin{equation}\label{asympsofphi2}
f(\Omega) \ \ = \ \ \sum_{(m_1,n_1),(m_2,n_2)\, \in  \,  (\Z\times \Z)'/\pm} \phi\(\ttwo{m_1}{n_1}{m_2}{n_2},\Omega\) ,
\end{equation}
where
 \begin{equation}\label{shadow3}
   (\Z\times \Z)'/\pm     \ \ = \ \    \{(0,1)\}\sqcup \{(c,d)\,|\,c>0, \gcd(c,d)=1\}\,.
\end{equation}

We shall later derive expressions for the Fourier modes of (\ref{poincare7new}) in terms of the Fourier transform of $h$,
\begin{equation}\label{lem:fhat}
    \aligned
    \widehat{h}(r) \ \ & = \ \ \int_\IR h(u)\,e^{-2\pi i ru}\,du \\
    & = \ \ 2\,\left({10\over\pi ^2 r^2}+1\right)  K_0(2 \pi|r|)
      \ + \
 {4 \left(3 \pi ^2 r^2+5\right)\over \pi^3|r|^3} K_1(2 \pi  |r|) \\ & \qquad\qquad  - \
 {32\over 3\pi \sqrt{|r|}}K_{7\over2}(2\pi |r|) \,, \ r\,\neq\,0\,,\\
 \widehat{h}(0) \ \ & = \ \ \f16\, .
    \endaligned
\end{equation}
This computation may be performed in a variety of ways:~shifting the contour of integration and wrapping around the branch cuts, using properties (\ref{branchposneg}); explicitly computing the Fourier transform  $h(u)(1+u^2)^{-s}$ for $\Re{s}$ large, and then analytically continuing to $s=0$; or taking the Fourier transform of the differential equation (\ref{poincare4}), and explicitly solving the resulting differential equation subject to the  constraints that $\hat{h}$ is continuous and has rapid decay.
\subsection{Fourier coefficients via a term-by-term analysis}
\label{sec:fouriermeth}

Returning to the expression (\ref{asympsofphi2}), we break the sum into pieces defined by
\begin{equation}\label{shadow4}
 f(\Omega) \ \  = \ \   \Sigma^{0,0}(\Omega) \  + \  \Sigma^{0,1} (\Omega)   \ + \   \Sigma^{1,0}(\Omega) \  + \   \Sigma^{1,1}(\Omega) \,,
\end{equation}
where
\begin{equation}\label{Sigma00def}
  \Sigma^{0,0} (\Omega) \ \ := \ \ \phi\(\ttwo{0}{1}{0}{1},\Omega\)  \ \  = \ \  {2\,\zeta(3)^2\over3}\,  y^3\,,
\end{equation}
\begin{equation}\label{Sigma01def}
  \Sigma^{0,1}(\Omega)  \ \ := \ \ \sum_{m_2\,=\,1}^\infty \sum_{\gcd(n_2,m_2)\,=\,1} \phi\(\ttwo{0}{1}{m_2}{n_2},\Omega\),
\end{equation}
\begin{equation}\label{Sigma10def}
  \Sigma^{1,0}(\Omega)  \ \ := \ \ \sum_{m_1\,=\,1}^\infty \sum_{\gcd(n_1,m_1)\,=\,1} \phi\(\ttwo{m_1}{n_1}{0}{1},\Omega\),
\end{equation}
and
\begin{equation}\label{Sigma11def}
  \Sigma^{1,1} (\Omega) \ \ := \ \ \sum_{m_1\,=\,1}^\infty \sum_{\gcd(n_1,m_1)\,=\,1}\sum_{m_2\,=\,1}^\infty \sum_{\gcd(n_2,m_2)\,=\,1}  \phi\(\ttwo{m_1}{n_1}{m_2}{n_2},\Omega\).
\end{equation}
The contribution of these terms to the Fourier modes, $\widehat{f}_n(y)$,   will be described in the rest of this appendix.  In sections~\ref{sec:Sigma01} and \ref{sec:Sigma11} we will   derive Fourier expansions
\be
\Sigma^{0,1} (\Omega) \ \  = \ \  \Sigma^{1,0}(\Omega) \ \  = \ \    \sum_{n\,\in\,\Z} e^{2 \pi i n x}\, \widehat{\Sigma}^{0,1}_n(y)\,,
\label{sigma01modes}
\ee
and
 \be
\Sigma^{1,1} (\Omega)   \ \ =  \ \  \sum_{  n_1\,\in\,\Z}\,\sum_{  n_2\,\in\,\Z} \, e^{2 \pi i( n_1 + n_2) x}\, \widehat{\Sigma}^{1,1}_{  n_1,  n_2}(y)\,,
\label{sigma11modes}
\ee
respectively.  The Fourier modes $\widehat{\Sigma}^{0,1}_n(y)$ and $\widehat{\Sigma}^{1,1}_n(y)$  are related to the $\widehat{f}_{n_1,n_2}$ of section~\ref{sec:boundaryterms} by the formulas
\begin{equation}\label{fhatsigmahatrelations}
\aligned
    \widehat{f}_{0,0}(y) \ \ & = \ \   \Sigma^{0,0} (\Omega)  \ + \ 2\,\widehat{\Sigma}^{0,1}_0(y) \ + \ \widehat{\Sigma}^{1,1}_{0,0}(y)\,, \\
    \widehat{f}_{n,0}(y) \ \ & = \ \ \widehat{f}_{0,n}(y) \ \ = \ \  \widehat{\Sigma}^{0,1}_n(y) \ + \ \widehat{\Sigma}^{1,1}_{n,0}(y)\,,\ \ n\,\neq\,0\,,\\
   \text{and~} \ \  \widehat{f}_{n_1,n_2}(y) \ \ & = \ \ \widehat{\Sigma}^{1,1}_{n_1,n_2}(y)\,,\ \ n_1,n_2\,\neq\,0\,.
\endaligned
\end{equation}
Writing \eqref{f00} as $ \widehat{f}_{0,0}(y)=  \widehat{f}^{(1)}_{0,0}(y)+ \widehat{f}^{(2)}_{0,0}(y)+ \widehat{f}^{(3)}_{0,0}(y)$, where $ \widehat{f}^{(r)}_{0,0}(y)$ is proportional to $y^{5-2r}$, we see that
\be
 \widehat{f}^{(1)}_{0,0}(y)   \ \ =  \ \ \Sigma^{0,0} (\Omega) \ \  = \ \  {2\,\zeta(3)^2\over3}\,y^3   \,.
\label{sigma01modes}
\ee

\subsubsection{Poisson summation for $\Sigma^{0,1}$ and $\Sigma^{1,0}$}
\label{sec:Sigma01}

We will now apply Poisson summation to ~\eqref{Sigma01def}  to put it  in the form (\ref{sigma01modes}).  First, reindexing the sum shows
\begin{equation}\label{Sigma01a}
\aligned
  &\Sigma^{0,1}(\Omega)  \ \  = \ \    \sum_{m_2\,=\,1}^\infty \sum_{n_2\,\in\,(\Z/m_2\Z)^*}\,\sum_{r\,\in\,\Z}
\phi\(\ttwo{0}{1}{m_2}{n_2+rm_2},\Omega\) \\
& = \ \
\sum_{m_2\,=\,1}^\infty \sum_{n_2\,\in\,(\Z/m_2\Z)^*}\sum_{r\in\Z}\, {1\over
  m_2^3}\,F((m_2x+n_2+rm_2+im_2 y)^{-1})\\
& = \ \
\sum_{m_2\,=\,1}^\infty \sum_{n_2\,\in\,(\Z/m_2\Z)^*}\,\sum_{r\in\Z} \,{4\,\zeta(3)^2\over
  m_2^3}\,h\left(m_2x+n_2+rm_2\over m_2 y\right).
\endaligned
\end{equation}
 The sum over $n_2$ ranges over all residue classes of integers modulo $m_2$ that are coprime to $m_2$.
Applying  Poisson summation to the inner sum,
\begin{eqnarray}
  \sum_{r\in\Z}   h\left(m_2x+n_2+rm_2\over m_2 y\right)&=&
  \sum_{n\in\Z}\int_\IR e^{-2\pi i n r} \, h\left(m_2x+n_2+rm_2\over m_2
    y\right)  \, dr\cr
&=& \sum_{n\in\Z}e^{2 \pi  i n {m_2x+n_2\over m_2} }  y\,\widehat{h}(ny)\,,
\end{eqnarray}
we obtain
\begin{equation}
  \label{eq:6}
\Sigma^{0,1}(\Omega) \ \ = \ \   \sum_{m_2\,=\,1}^\infty  {4\,\zeta(3)^2\over m_2^3}\,\sum_{n_2\,\in\,(\Z/m_2\Z)^*}\,\sum_{n\,\in\,\Z}
 e^{2\pi i n (x+{n_2\over m_2})}  \, y \, \widehat{h}(n y)\,.
\end{equation}
 The $n_2$-sum can be written in terms of the {\em  Ramanujan sum }
  \be
 \label{deframan}
  c_{m_2}(n) \ \ = \ \  \sum_{n_2=1\atop (m_2,n_2)=1}^{m_2} e^{2 \pi i \tfrac{n_2}{m_2} n} \,,
  \ee
  which satisfies the identity
\be
\sum_{m_2=1}^\infty \frac{c_{m_2}(n)}{m_2^{r}}  \ \ = \ \   \frac{\sigma_{1-r}(|n|)}{ \zeta(r)}\,.
\label{ramanid}
\ee
 When applied in the  $r=3$ case to \eqref{eq:6},  this results in the expression
\begin{equation}\label{Sigma01final}
  \Sigma^{0,1} (\Omega) \ \  = \ \    \Sigma^{1,0}(\Omega) \ \    = \ \  4\,\zeta(3) \sum_{n\,\in\,\Z} e^{2\pi i nx} \,\sigma_{-2}(|n|)\,  y\,\widehat{h}(n y)\,,
\end{equation}
which has the form~\eqref{sigma01modes}.
Formula \eqref{ramanid} also applies when $n=0$ provided we use the convention that $\sigma_{-2}(0) = \zeta(2)$.
  Since $\widehat{h}(0)=1/6$ by (\ref{lem:fhat}), it follows that the $n=0$ term in~\eqref{Sigma01final}   gives the total contribution
\be
\widehat{\Sigma}^{0,1}_0(y) \ + \ \widehat{\Sigma}^{1,0}_0(y) \ \ = \ \  \smallf{4}{3}\, \zeta(2)\,\zeta(3)\, y
\ \  =  \ \ \widehat{f}^{(2)}_{0,0}
\label{zerooneterm}
\ee
to $\widehat{f}_{0,0}(y)$,
which is the second term on the right hand side of~\eqref{f00}.
Furthermore,
\begin{equation}\label{fhatmanip}
\widehat{\Sigma}^{0,1}_n(y) \ \ = \ \  4\,\zeta(3)\, \sigma_{-2}(|n|) \, y\,\widehat{h}(n y) \ \ \ \  \text{for} \ \ n\,\neq\,0\,.
\end{equation}

\subsubsection{Poisson summation for $\Sigma^{1,1}$}\label{sec:Sigma11}

In order to produce the Fourier series~\eqref{sigma11modes},
 we perform a double Poisson summation on  the definition of $\Sigma_{1,1}$ in
 \eqref{Sigma11def} to obtain the formula
\begin{multline}\label{Sigma11a}
   \Sigma^{1,1}(\Omega)  \\    = \ \
\sum_{m_1\,=\,1}^\infty \sum_{n_1\,\in\,(\Z/m_1\Z)^*}\,\sum_{r_1\,\in\,\Z}
\sum_{m_2\,=\,1}^\infty \sum_{n_2\,\in\,(\Z/m_2\Z)^*}\,\sum_{r_2\,\in\,\Z}
 \phi\(\ttwo{m_1}{n_1+r_1m_1}{m_2}{n_2+r_2m_2},\Omega\)
\\  =
\ \
\sum_{m_1,m_2\,=\,1}^\infty \sum_{\srel{n_1\,\in\,(\Z/m_1\Z)^*}{\srel{n_2\,\in\,(\Z/m_2\Z)^*}{\hat{n}_1,\hat{n}_2\,\in\,\Z}}}\,
\int_{\IR^2}\phi\(\ttwo{m_1}{n_1+r_1m_1}{m_2}{n_2+r_2m_2},\Omega\) e^{-2\pi i \left(\hat{n}_1r_1+\hat{n}_2 r_2\right)}\,dr_1\,dr_2\,.
\end{multline}
The integral is given by
\begin{equation}\label{Sigma11b}
  \int_{\IR^2} F(\smallf{m_1(x+i  y)+(n_1+r_1m_1)}{m_2(x+i  y)+(n_2+r_2m_2)})\,
  \f{e^{-2\pi i \left(\hat{n}_1r_1+\hat{n}_2 r_2\right)}}{ |m_1(n_2+r_2m_2)-m_2(n_1+r_1m_1)|^{3}}\,dr_1\,dr_2\,.
\end{equation}
With the change of variables $r_1\mapsto r_1-x-\f{n_1}{m_1}$, $r_2\mapsto r_2-x-\f{n_2}{m_2}$ this integral becomes
\begin{multline}\label{Sigma11c}
  e^{2\pi i\left((\hat{n}_1+\hat{n}_2)x+\hat{n}_1\smallf{n_1}{m_1}+\hat{n}_2\smallf{n_2}{m_2}\right)} \\ \times
\int_{\IR^2} |m_1m_2(r_2-r_1)|^{-3}\,F(\smallf{m_1(r_1+i y)}{m_2(r_2+i y)})\,e^{-2\pi i \left(\hat{n}_1r_1+\hat{n}_2 r_2\right)}\,dr_1\,dr_2 \\ = \ \
 e^{2\pi i\left((\hat{n}_1+\hat{n}_2)x+\hat{n}_1\smallf{n_1}{m_1}+\hat{n}_2\smallf{n_2}{m_2}\right)}
  \,  y^{-1} \, |m_1m_2|^{-3}\\
 \times
\int_{\IR^2} |r_2-r_1|^{-3}\,F(\smallf{r_1+i}{r_2+i})\,e^{-2\pi i\left( \hat{n}_1 yr_1+\hat{n}_2 y r_2\right)}\,dr_1\,dr_2\,,
\end{multline}
after changing variables $r_1\mapsto  yr_1$ and $r_2\mapsto  yr_2$.
 After applying (\ref{ramanid}) twice,  we can write (\ref{Sigma11a}) in the form~\eqref{sigma11modes} with mode coefficients given by
\be
\widehat{\Sigma}^{1,1}_{\hat n_1,\hat n_2}(y) \  \ = \ \
  \f{4}{y} \, \sigma_{-2}(|\hat{n}_1|)\,\sigma_{-2}(|\hat{n}_2|)\,   {\mathcal I}(\hat{n}_1,\hat{n}_2; y)\,\,
\label{modesrep}
\ee
where
\be
 {\mathcal I}(\hat{n}_1,\hat{n}_2; y) \ \ = \ \ {1\over4\,\zeta(3)^2}\int_{\IR^2} \f{F(\smallf{r_1+i}{r_2+i})}{|r_2-r_1|^3}\,e^{-2\pi i\left(\hat{n}_1 yr_1+\hat{n}_2 y r_2 \right)}\,dr_1\,dr_2\,.
\label{intdefs}
\ee
Using
\begin{equation}\label{Sigma11e}
 F(\smallf{r_1+i}{r_2+i})
 \ \ = \ \   4\,\zeta(3)^2\,h(\smallf{r_1r_2+1}{r_2-r_1})\,,
\end{equation}
this can be rewritten as
\begin{equation}\label{Sigma11f}
 {\mathcal I}(\hat{n}_1,\hat{n}_2; y) \ \ := \ \  \int_{\IR^2}
 \f{h(\smallf{r_1r_2+1}{r_2-r_1})}{|r_2-r_1|^3}\ e^{-2\pi i
   \left(\hat{n}_1 yr_1+\hat{n}_2 y r_2\right)}\,dr_1\,dr_2\,.
\end{equation}
\begin{lem}\label{lem:integrability}
The integral (\ref{Sigma11f}) is absolutely convergent.  Consequently, $\mathcal{I}(\hat{n}_1,\hat{n}_2; y)$ is bounded in $ y$  and the Fourier modes (\ref{modesrep}) are $O(y^{-1})$.
\end{lem}
\begin{proof}
 Since the exponential has modulus 1 and $h\ge 0$, it suffices to show the convergence of ${\mathcal I}(0,0;y)$.  Recalling that  $h(u)$ is bounded by a constant multiple of $\sigma(u)=(u^2+1)^{-3/2}$, this reduces to the convergence of the integral
 \begin{multline}\label{integrability1}
    \int_{\R^2}|r_2-r_1|^{-3}\,(1+(\smallf{r_1 r_2+1}{r_2-r_1})^2)^{-3/2}\,dr_1\,dr_2 \ \ =
     \\ \int_{\R^2}(1+r_1^2)^{-3/2}\,(1+r_2^2)^{-3/2}\,dr_1\,dr_2 \,,
 \end{multline}
 which is of course finite.
\end{proof}
\noindent
In fact, we know from \eqref{modesrep} and \eqref{f00} that
\be
\label{zeromodess}
\widehat{\Sigma}^{1,1}_{0,0}(y) \ \ = \ \   {4\over y} \,\zeta(2)^2 \,  {\mathcal I}(0,0; y) \  \ = \ \  \widehat f^{(3)}_{0,0}(y) \ \ = \  \  \f{4\,\zeta(4)}{y}\,,
\ee
so that ${\mathcal I}(0,0; y)=2/5$  (this can also be verified  by   numerical integration).

  Equations (\ref{zerooneterm})-(\ref{fhatmanip}) show that $\widehat{\Sigma}_n^{0,1}(y)=O(y)$ for small $y$, and so (\ref{fhatsigmahatrelations}) and  lemma~\ref{lem:integrability} imply  that the modes $\widehat{f}_{n_1,n_2}$ are at most $O(1/y)$ for small
values of $y$.
Nevertheless, the $y\to 0$ limit of this expression is more singular than $1/y$ because the $n$-th Fourier coefficient $\widehat{f}_{n}(y)$ includes    the sum
\be \sum_{n_1\,\neq\,0,n}
   \widehat{f}_{n_1,n-n_1}(y) \ \ = \ \   \f{4}{y} \, \sum_{n_1\,\neq\,0,n}\,\sigma_{-2}(|n_1|)\,
     \sigma_{-2}(|n-n_1|)\, {\mathcal I}(n_1,n-n_1; y)\,.
\label{partialsum}
\ee
Indeed, in   (\ref{nonzerolimit})  the $y\to 0$ behaviour of this expression  was shown to be proportional to a constant multiple of $1/y^2$.

\section{ Poincar\'e series and  Eisenstein automorphic distributions}
\label{sec:schmidmethod}
We now return to the sum over
  $\gamma$ in~\eqref{poincare7new}\,,
\begin{equation}\label{poincare8}
\sum_{\g\,\in\,\mathcal S}
  (\det \g)^{-3} F(\g \Omega)\,,
\end{equation}
where we recall that ${\mathcal S}=\{\pm 1\}\backslash \{\ttwo{m_1}{n_1}{m_2}{n_2}\in M_2(\Z)\cap GL^+(2,\IR)|\gcd(m_1,n_1)=\gcd(m_2,n_2)=1\}$.
We begin with some comments about the structure of the set $\mathcal S$.
Suppose $[m\,n]\in \Z^2$ is a vector with $d=\gcd(m,n)$. Then $d$ divides $[m\,n]\g$ for any integral $2\times 2$ matrix $\g$.  Consequently, if $\g\in SL(2,\Z)$  then $\gcd([m\,n])=\gcd([m\,n]\g)$.  If $\g$ has the form $\g=\ttwo{n_2}{b}{-m_2}{d}\in SL(2,\Z)$ and $\ttwo{m_1}{n_1}{m_2}{n_2}\in\mathcal S$, then $\ttwo{m_1}{n_1}{m_2}{n_2}\g = \ttwo{p_1}{p_2}{0}{1}$ for some relatively prime integers $p_1,p_2\in \Z$, where $p_1=m_1n_2-n_1m_2>0$.    Thus
\begin{equation}\label{poincare9}
 { \mathcal S}    \ \ = \ \   \left\{
 \ttwo{p_1}{p_2}{0}{1}\g \,|\, p_1>0, p_2\in (\Z/p_1\Z)^*, \g \in \G
 \right\}
\end{equation}
parameterizes elements of $\mathcal S$ via $\G=PSL(2,\Z)$.

In light of (\ref{poincare9}),
(\ref{poincare8}) becomes
\begin{equation}\label{poincare10a}
\sum_{\g\,\in\,\mathcal S}
  (\det \g)^{-3} F(\g \Omega)  \ \ = \ \  \sum_{p_1\,=\,1}^\infty\sum_{p_2\,\in\,(\Z/p_1\Z)^*} \sum_{\g\,\in\,PSL(2,\Z)}p_1^{-3}\,F\(\ttwo{p_1}{p_2}{0}{1}\g \Omega\).
\end{equation}
By virtue of its  definition in (\ref{poincare6}), $F\(\ttwo{p_1}{0}{0}{1}\Omega\)=F(p_1\Omega)=F(\Omega)$ and so  (\ref{poincare10a}) can be written as
\begin{multline}\label{poincare10b}
  \sum_{p_1\,=\,1}^\infty\sum_{p_2\,\in\,(\Z/p_1\Z)^*} \sum_{\g\,\in\,PSL(2,\Z)}p_1^{-3}\,F\(\ttwo{1}{p_2/p_1}{0}{1}\g \Omega\) \\
  = \ \
 \sum_{\f{p_2}{p_1}\,\in\,\Q} p_1^{-3} \sum_{\g\,\in\,\G_\infty\backslash \G} F(\smallf{p_2}{p_1}+\g\Omega) =  \sum_{\g\,\in\,\G_\infty\backslash \G} \Phi(\g\Omega)\,,
\end{multline}
where
\begin{equation}\label{poincare12a}
  \Phi(\Omega)  \  \ = \ \    \sum_{\f{p_2}{p_1}\,\in\,\Q} p_1^{-3}F(\smallf{p_2}{p_1}+\Omega)\,.
\end{equation}
Thus (\ref{poincare10b}) writes the sum (\ref{poincare8}) as a sum of left translates of $\Phi$ over $\G_\infty\backslash \G$, the type of the sum that the terminology ``Poincar\'e series'' is traditionally reserved for.

A standard double coset decomposition for $\G_\infty\backslash \G/\G_\infty$ and application of Poisson summation (see~\cite{sarnak,iwaniec}) to the last expression in (\ref{poincare10b}) gives
\begin{multline}\label{poincare14}
  \sum_{\g\,\in\,\mathcal S}
  (\det \g)^{-3} F(\g \Omega) \ \ = \ \  \Phi(\Omega) \ + \\ \sum_{c\,=\,1}^\infty \sum_{d\,\in\,(\Z/c\Z)^*}\,\sum_{n\,\in\,\Z}e^{2\pi in(x+d/c)}
   \int_{\IR}e^{-2\pi i nr}\,
   \Phi(\smallf ac - \smallf{1}{c^2(r+i y)})\,dr\,.
\end{multline}
To compute this integral, we use the following Fourier expansion of $\Phi$:

\begin{lem}\label{schmid}  In terms of  the function $h$ from (\ref{poincare6}),
\begin{equation}\label{schmidformula}
\Phi(x+i y)  \ \  = \ \   4\,\zeta(3)\,\sum_{n\,\in\,\Z} \sigma_{-2}(|n|) \, e^{2\pi i nx} \, y\,\widehat{h}(n  y)\,,
\end{equation}
where  $\hat{h}(\cdot)$ was computed in (\ref{lem:fhat}) and $\sigma_{-2}(|n|)=\sum_{d|n}d^{-2}$ is to be interpreted as $\zeta(2)$ when $n=0$.
\end{lem}
\begin{proof}
Writing the rationals in (\ref{poincare12a}) as an integer plus a rational in the interval $[0,1)$ we have that
\begin{equation}\label{schmidpf1}
\aligned
  \Phi(x+i y) \ \  & = \ \ \sum_{p_1\,=\,1}^\infty p_1^{-3} \, \sum_{p_2\,\in\,(\Z/p_1\Z)^*} \sum_{n\,\in\,\Z} F(\smallf{p_2}{p_1} + x+n+i y)
  \\
  & = \ \ \sum_{p_1\,=\,1}^\infty p_1^{-3} \, \sum_{p_2\,\in\,(\Z/p_1\Z)^*} \sum_{n\,\in\,\Z}
  \int_{\IR} F(\smallf{p_2}{p_1} + x+u+i y)\,e^{-2\pi i nu}\,du \\
  & = \ \ \sum_{p_1\,=\,1}^\infty p_1^{-3} \!\!\!  \sum_{p_2\,\in\,(\Z/p_1\Z)^*} \sum_{n\,\in\,\Z} e^{2\pi i n(x+\smallf{p_2}{p_1})}\int_{\IR} F(u+i y)\,e^{-2\pi i nu}\,du \\
  & = \ \ \sum_{n\,\in\,\Z}e^{2\pi i nx}  \( \sum_{p_1\,=\,1}^\infty p_1^{-3} \, \sum_{p_2\,\in\,(\Z/p_1\Z)^*}
  e^{2\pi i np_2/p_1}\) \\ & \qquad\qquad\qquad\qquad\qquad\qquad \times \ 4\,\zeta(3)^2\,\int_{\IR} h(\smallf{u}{ y})\,e^{-2\pi i nu}\,du
  \endaligned
\end{equation}
after applying Poisson summation and (\ref{poincare6}).  The lemma now follows from (\ref{ramanid}).
\end{proof}
After inserting (\ref{schmidformula}), the integral in  (\ref{poincare14}) becomes
\begin{equation}\label{poincare15new}
   4\,\zeta(3) \int_{\IR}e^{-2\pi i nr}\,\sum_{m\,\in\,\Z} \sigma_{-2}(|m|) \,e^{2\pi i \smallf{ma}{c}-2\pi i \smallf{mr}{c^2(r^2+y^2)}}\,
    \smallf{y}{c^2(r^2+y^2)}\,\hat{h}( \smallf{my}{c^2(r^2+y^2)})\,dr\,.
\end{equation}
In terms of the  Kloosterman sum $S(a,b;c):=\sum_{x\in (\Z/c\Z)^*}e((ax+bx^{-1})/c)$, the second term on the righthand side of  (\ref{poincare14}) equals
\begin{multline}\label{poincare16new}
   4\,\zeta(3)\, \sum_{c\,=\,1}^\infty\sum_{m,n\,\in\,\Z}e^{2\pi i nx}\,\sigma_{-2}(|m|)\,S(m,n;c)
    \\ \times \ \int_{\IR}e^{-2\pi i \(nr+\smallf{mr}{c^2(r^2+y^2)}\)}\smallf{y}{c^2(r^2+y^2)}\,\hat{h}( \smallf{my}{c^2(r^2+y^2)})\,dr\,.
\end{multline}
The Fourier coefficient of the Fourier mode $x\mapsto e(nx)$ of (\ref{poincare8}) can thus be read off from this and (\ref{poincare14}); combining with (\ref{poincare7new})  proves the following
\begin{lem}\label{kloostermanfourierexpansion}
The Fourier modes $\widehat{f}_n(y)$ from (\ref{fouriermodes}) are given as the sums
\begin{multline}\label{kloosterfourier}
    \widehat{f}_n(y) \ \ = \ \ \smallf{2\,\zeta(3)^2}{3}{\mathcal F}_{n,3}(y) \ + \ 4\,\zeta(3)\,\sigma_{-2}(|n|)\,y\,\widehat{h}(ny) \ + \\ 4\zeta(3)\sum_{\srel{c\,>\,0}{m\,\in\,\Z}}\sigma_{-2}(|m|)S(m,n;c)
    \int_{\IR}e^{-2\pi i r \(n+\smallf{m}{c^2(r^2+y^2)}\)}\smallf{y}{c^2(r^2+y^2)}\,\hat{h}\( \smallf{my}{c^2(r^2+y^2)}\)dr,
\end{multline}
where ${\mathcal F}_{n,s}(y)$ is defined in (\ref{eisenzero})-(\ref{nonzeroeisen}).
\end{lem}

\medspace

\noindent
{\bf Remark:}
The function $\Phi(x+iy)$ can be interpreted in terms of
Schmid's automorphic Eisenstein distribution  as
\begin{equation}\label{autodist1}
   \Phi(x+iy) \ \ = \ \  4\,\zeta(3)\,\int_{\R}\tau_{2}(u+x)\,h(\smallf uy)\,du\,,
\end{equation}
where
\begin{equation}\label{autodist2}
    \tau_\nu(u) \ \  =
    \ \ \sum_{\srel{p,q\,\in\,\Z}{q>0}}q^{-\nu-1}\,\d_{u=p/q}  \ \ = \ \ \sum_{n\,\in\,\Z}\sigma_{-\nu}(|n|)\,e^{2\pi i nu}
\end{equation}
is the restriction to $\R$ of the automorphic distribution corresponding to the Eisenstein series $E_{(\nu+1)/2}$ (see \cite[\S4]{korea}). It can be alternatively be thought of as a distributional limit
of values of $E_{(\nu+1)/2}(x+iy)$ as $y\to 0$ \cite{flato}.  When multiplied by a suitable power of $y$, integrals such as (\ref{autodist1}) represent the embedding of  vectors (playing the role of $h$) in the line model of a principal series representation of $SL(2,\R)$, into spaces of automorphic functions on the group $SL(2,\R)$.  Through this formalism, $\Phi(x+iy)+\f{2\zeta(3)^2}{3}y^3$ (whose periodization over $\g\in \G_\infty\backslash\G$ is $f(x+iy)$ in (\ref{poincare7new})) is itself naturally related to the restriction of  an automorphic function on $SL(2,\R)$ to its upper triangular subgroup $\{\ttwo 1x01 \ttwo{y^{1/2}}{0}{0}{y^{-1/2}}|x\in\R,y>0\}$.
Furthermore, identity (\ref{schmidformula}) is an immediate consequence of (\ref{autodist2}).

\section{Solution as expressed via spectral theory}
\label{sec:spectral}

  We shall now present the spectral expansion of the solution $f$ to (\ref{laplaceeigen}).  Though we shall not directly link it to the expression (\ref{kloosterfourier}), there is in fact a famous connection between Kloosterman sums and the spectral theory of automorphic forms (see, for example,~\cite{sarnak,iwaniec2}).

 A significant complication here is that the source term   $-(2\zeta(3)E_{\frac32})^2$ in (\ref{laplaceeigen})
  is not integrable over the automorphic quotient, hence a divergent contribution must be subtracted from it.
    However, it is striking that the corresponding source term {\em is}  integrable for $D_5$, $E_6$, $E_7$,
    and $E_8$~\cite{GMRV,GMV}.

Consider $L^2(\G\backslash \U)$, with the inner product
\be\label{innerproduct}
\langle f_1, f_2 \rangle \ \  = \ \  \int_{\G\backslash {\mathbb H}}f_1(x+iy) \, \overline{f_2(x+iy)}\,\f{dxdy}{y^2}\,.
\ee
The R\"olcke-Selberg spectral expansion theorem  states that any function $H\in L^2(\G\backslash \U)$ can be expanded as
\begin{equation}\label{sl2spectexp}
    H(z) \ \   =  \ \  \sum_{j\,=\,0}^\infty \langle H,\phi_j\rangle\,\phi_j(\Omega) \
+ \ \f{1}{4\pi i}\int_{\Re{s}=1/2}\langle H,E_s\rangle E_s(\Omega)\,ds\,,
\end{equation}
where the $\phi_j$ are an orthonormal basis for the discrete spectrum satisfying $(\Delta+\l_j)\phi_j=0$ for some eigenvalue $\l_j\ge 0$ ($\phi_0$ is the constant $\sqrt{\frac{3}{\pi}}$, whereas the $\phi_j$ for $j\ge 1$ are Maass cusp forms).

Consider the inhomogeneous differential equation
\begin{equation}\label{sl2laplaceA1}
    (\Delta-12)F  \ \  = \ \   -E_{3/2}^2\,, \ \ F\ \text{automorphic under}\ \G\,,
\end{equation}
which differs from (\ref{laplaceeigen}) by the constant multiple $4\zeta(3)^2$.  The polynomially-bounded solutions to its
  homogeneous analog
\begin{equation}\label{sl2laplaceA2}
    (\Delta-12)F   \ \ = \ \   0
\end{equation}
which are automorphic under $\G$ are all scalar multiples of the
 Eisenstein series $E_4$, which we recall grows like $y^4$ as $y\rightarrow \infty$.
Because of the growth condition (\ref{largebound}),
 the solution $f$ to (\ref{laplaceeigen}) is $4\zeta(3)^2$ times the unique
   $\G$-automorphic solution $F(x+iy)$ to (\ref{sl2laplaceA1}) which grows by at most $O(y^3)$ in the cusp.

For large values of $y$, the Eisenstein series $E_s(x+iy)$ is asymptotic to its constant term
 ${\mathcal F}_{0,s}(y)=y^s+c(2s-1)y^{1-s}$ given in (\ref{eisenzero}), where $c(s)=\f{\xi(s)}{\xi(s+1)}$ and $\xi(s)=\pi^{-s/2}\G(s/2)\zeta(s)$.
  Thus the righthand side of (\ref{sl2laplaceA1})  is asymptotic to
   $-y^3 - 2c(2)y - c(2)^2y^{-1}$ as $y\rightarrow \infty$.
    The automorphic function $\hat{E}_1 $, defined as the constant term in
    the Laurent expansion
\begin{equation}\label{E1Laurent}
  E_s   \ \ = \ \   \smallf{3}{\pi(s-1)} \ + \ \hat{E}_1  \ + \ (s-1)\,\hat{\hat{E}}_1  \ + \ O((s-1)^2)
\end{equation}
 of $E_s$ at $s=1$, is asymptotic to $y$ as $y\rightarrow \infty$.  Therefore
\begin{equation}\label{rhsmanip1}
  H \ \ := \ \   -E_{3/2}^2 \ + \ E_3 \ + \ 2\,c(2)\, \hat
    {E}_1 \ \   = \ \  O(\log(y))
\end{equation}
and is in particular in $L^2(\G\backslash \U).$  Applying the spectral expansion (\ref{sl2spectexp}), we see that
\begin{equation}\label{rhsmanip2}
    (\Delta-12)^{-1}H  =  -\, \sum_{j\,=\,0}^\infty  \f{\langle H,\phi_j\rangle}{\l_j+12}\,\phi_j(\Omega) \
 + \ \f{1}{4\pi i}\int_{\Re{s}=1/2}\f{\langle H,E_s\rangle}{s(s-1)-12} E_s(\Omega)\,ds\,.
\end{equation}
Since $H$ is square integrable and since the coefficients on the righthand side of (\ref{rhsmanip2}) are each smaller than their respective counterparts in (\ref{sl2spectexp}), Parseval's theorem shows that $(\Delta-12)^{-1}H$
 is also in $L^2(\Gamma\backslash \mathbb H)$.

 Applying $\Delta$ to both sides of (\ref{E1Laurent}) and comparing constant terms in $s$ at $s=1$ results in the differential equation $\Delta\hat{E}_1=\f{3}{\pi}$.  Since $\Delta E_3=6E_3$, we conclude
 \begin{multline}\label{theinhomogsolintermsofinnerproducts}
 F \ \   =  \ \  \smallf{1}{6}E_3(\Omega)+2c(2)(\smallf{1}{12}\hat{E}_1(\Omega)+\smallf{1}{48\pi}) \
-\ \sum_{j\,=\,0}^\infty \f{\langle H,\phi_j\rangle}{\l_j+12}\,\phi_j(\Omega) \\ +
\ \f{1}{4\pi i}\int_{\Re{s}=1/2}\f{\langle H,E_s\rangle}{s(s-1)-12} E_s(\Omega)\,ds
\end{multline}
 is the unique solution to (\ref{sl2laplaceA1}) which is $O(y^3)$ for large $y$.

The solution (\ref{theinhomogsolintermsofinnerproducts}) can be explicated by computing the inner products $\langle H,\phi_j\rangle$ and $\langle H,E_s\rangle$.  For $j\ge 1$ the former are more complicated and can be computed in terms the $L$-functions of the Maass forms $\phi_j$ using the Rankin-Selberg unfolding method.  However, since the Maass forms themselves are quite mysterious, this is of dubious direct utility.  The Maass forms are characterized by having zero constant term.  This indicates that the nonconstant Fourier modes of the solution to (\ref{laplaceeigen}) are difficult to directly compute using the spectral expansion.

At the same time, the inner products $\langle H,\phi_0\rangle$ and   $\langle H,E_s\rangle$ can be computed very explicitly, and together give an alternative derivation of the constant Fourier mode (\ref{zeromode})  of the solution to (\ref{laplaceeigen}).  The rest of this appendix indicates how these computations are carried out.

Let ${\mathcal F}_C:=\{x+iy| x\in[-\frac12,\frac12], x^2+y^2\geq1,
y\leq C\} $ denote the points in the standard fundamental domain
$\mathcal F$ for $SL(2,\Z)$ having imaginary part bounded by $C$.
Let $\L^C$ be the truncation operator on automorphic functions which
subtracts the constant term at points in the $\G$-translates of
$\mathcal F-\mathcal F_C$:
\begin{equation}\label{truncationoperator}
    (\L^C \phi)(x+iy) \ \ = \ \ \left\{
                                \begin{array}{ll}
                                  \phi(x+iy)-\int_0^1 \phi(u+iy)du, & y\,\in\,{\mathcal F}-{\mathcal F}_C\,,\\
                                  \phi(x+iy), & y\,\in\,{\mathcal F}_C\,
                                \end{array}
                              \right.
\end{equation}
(this formula defines $\L^C\phi$ on the fundamental domain ${\mathcal F}$; its value elsewhere is determined by automorphy).

   The {\em Maass-Selberg relations} state that
\begin{multline}\label{maassselbergrelations}
  \langle \L^C E_{s_1},E_{s_2}\rangle  \ \  = \ \
  \f{C^{s_1+\overline{s_2}-1}}{s_1+\overline{s_2}-1}+c(2\overline{s_2}-1)\f{C^{s_1-\overline{s_2}}}{s_1-
\overline{s_2}}
  +c(2s_1-1)\f{C^{\overline{s_2}-s_1}}{\overline{s_2}-s_1}\cr
+c(2s_1-1)c(2\overline{s_2}-1)
  \f{C^{1-s_1-\overline{s_2}}}{1-s_1-\overline{s_2}}\,.
\end{multline}
Since $\L^C E_{s_1}$ decays rapidly in the cusp, this inner product
differs from $\int_{\mathcal F_C}E_{s_1}\overline{E_{s_2}}\,\f{dxdy}{y^2}$ by an additive term of size $o(1)$ as $C\rightarrow\infty$.

Since $H$ is square-integrable its inner product with $E_{1/2+it}$ converges, though the inner products of its three constituent terms in (\ref{rhsmanip1}) do not.  Write
\begin{multline}\label{limit}
  \langle H,E_{1/2+it}\rangle   \ \ = \ \   -\lim_{C\rightarrow\infty}\int_{{\mathcal F}_C} E_{\frac32}^2(\Omega)
E_{\frac12-it}(\Omega)\,\smallf{dx\,dy}{y^2} \ + \\
  \lim_{C\rightarrow\infty}\int_{{\mathcal F}_C} E_3(\Omega)E_{\frac12-it}(\Omega)\,\smallf{dx\,dy}{y^2} \ +
 \  2\,c(2)\lim_{C\rightarrow\infty}\int_{{\mathcal F}_C} \hat{E}_1(\Omega)E_{\frac12-it}(\Omega)\,\smallf{dx\,dy}{y^2}\,.
\end{multline}
Since the compensating $o(1)$ terms disappear in the $C\to\infty$ limit, $E_{1/2-it}$ can be replaced with $\L^C E_{1/2-it}$ in each of
the above integrals.
The second integral on the righthand side can then be directly handled
by (\ref{maassselbergrelations}), while the third integral requires simply
  taking the constant term in the Laurent expansion of both sides at $s_2=1$.

However, the integral involving $E_{\frac32}(z)^2E_s(z)$ is more
subtle since standard regularization techniques (such as Zagier's
method~\cite{zagiermethod}) do not directly apply.
  This is because if we unfold the fastest growing series, $E_{\frac32}$, there still
remains the product $E_{\frac32}E_{\frac12-it}$, which has size on the order of  $y^2$.
 Zagier's method first truncates $E_{\frac32}$ by subtracting a term of size $y^{3/2}$,
and what remains decays only like $y^{-1/2}$; this is not enough to get integrability.
Instead we play a similar game as in (\ref{rhsmanip1}) by using other Eisenstein series
 which match those growth rates, and rewrite the first integral on the righthand side of (\ref{limit}) as
\begin{multline}\label{tripleprodint1}
  \int_{\mathcal F_C} E_{3/2}^2(\Omega)E_{\frac12-it}(\Omega)\,\smallf{dxdy}{y^2} \ \  = \\
  \int_{{\mathcal F}_C}
  E_{3/2}(\Omega)[E_{3/2}(\Omega)E_{\frac12-it}(\Omega)-E_{2-it}(\Omega)-\phi(\smallf12-it)E_{2+it}(\Omega)]\,\smallf{dxdy}{y^2}
  \cr
 +  \int_{{\mathcal F}_C} E_{\frac32}(\Omega)\,E_{2-it}(\Omega)\,\smallf{dxdy}{y^2}\ + \ \phi(\smallf12-it)
\int_{{\mathcal F}_C} E_{3/2}(\Omega)\,E_{2+it}(\Omega)\,\smallf{dxdy}{y^2}\,.
\end{multline}
The advantage of this rearrangement is that the bracketed
expression on the righthand side grows at most like $O(\log y)$, and so Zagier's truncation applies.  The last two integrals can be estimated using the Maass-Selberg relations (\ref{maassselbergrelations}).

Finally, the inner product $\langle H,\phi_0\rangle$ can be also
 computed using these techniques (using the fact it the residue of $E_s$ at $s=1$ is a constant function), though an important simplification
occurs because $\phi_0$ is constant:~namely, the truncated integral of
$E_{3/2}^2\phi_0=\f{\sqrt{3}}{\pi}E_{3/2}E_{3/2}$ can be computed using the Maass-Selberg relations (\ref{maassselbergrelations}).

\end{document}